\documentclass[a4paper,UKenglish,cleveref, autoref]{mylipics2019}
\title{Logarithmic Weisfeiler-Leman Identifies All Planar Graphs}

\author{Martin Grohe}{RWTH Aachen University, Aachen, Germany}{grohe@cs.rwth-aachen.de}{https://orcid.org/0000-0002-0292-9142}{}
\author{Sandra Kiefer}{University of Warsaw, Warsaw, Poland \and RWTH Aachen University, Aachen, Germany}{kiefer@cs.rwth-aachen.de}{https://orcid.org/0000-0003-4614-9444}{}

\authorrunning{M. Grohe and S. Kiefer} 

\Copyright{Martin Grohe and Sandra Kiefer} 



\funding{Sandra Kiefer's research was supported by the European Research Council under the European Unions Horizon 2020 research and innovation programme (ERC consolidator grant LIPA, agreement no.\ 683080).}

\nolinenumbers 

\usepackage{mathtools}
\usepackage{latexsym}
\usepackage{anyfontsize}
\usepackage{stmaryrd}\SetSymbolFont{stmry}{bold}{U}{stmry}{m}{n}
\usepackage[dvipsnames]{xcolor}
\usepackage{tikz}
\usetikzlibrary{arrows}
\usetikzlibrary{shapes.geometric}
\usetikzlibrary{calc}
\usetikzlibrary{decorations.pathmorphing}
\usetikzlibrary{shapes.geometric}
\usetikzlibrary{calc}
\tikzset{current point is local=true}

\newcommand{\case}[1]{\par\medskip\noindent\textit{Case #1: }}
\newenvironment{cs}{
  \begin{description}
    \renewcommand{\case}[1]{\item[\itshape\mdseries Case ##1:]}
  }{
  \end{description}
}
\newcommand{\uend}{\hfill$\lrcorner$}

\newenvironment{eroman}{\enumerate[(i)]}{\endenumerate}

\renewcommand{\mathbf}[1]{\boldsymbol{#1}}
\renewcommand{\max}{\operatorname{max}}
\renewcommand{\min}{\operatorname{min}}
\renewcommand{\phi}{\varphi}
\newcommand{\bigmid}{\;\big|\;}
\newcommand{\Bigmid}{\;\Big|\;}
\newcommand{\formel}[1]{\textsf{\upshape #1}}
\newcommand{\logic}[1]{\textsf{\upshape #1}}
\newcommand{\FOL}{\logic{FO}}
\newcommand{\ad}{\operatorname{ad}}
\newcommand{\width}{\operatorname{wd}}
\newcommand{\Nat}{{\mathbb N}}
\newcommand{\Real}{{\mathbb R}}
\newcommand{\bag}{\beta}
\newcommand{\CC}{\mathcal C}
\newcommand{\CM}{\mathcal M}
\renewcommand{\vec}{\overline}

\newcommand{\Bigllcurly}{\Big\{\hspace{-5pt}\Big\{}
\newcommand{\Bigrrcurly}{\mbox{$\Big\}\hspace{-5pt}\Big\}$}}
\newcommand{\atp}{\operatorname{atp}}

\DeclareMathOperator{\qd}{qd}
\newcommand{\WL}[2]{\logic{WL}^{#1}_{#2}}
\newcommand{\LC}[2]{\logic{C}^{#1}_{#2}}

\theoremstyle{plain}
\theoremstyle{definition}
\newcommand{\Fraisse}{Fra\"{\i}ss{\'e}}

\begin{document}
\maketitle

\begin{abstract}
The Weisfeiler-Leman (WL) algorithm is a well-known combinatorial
procedure for detecting symmetries in graphs and it is widely used in
graph-isomorphism tests. It proceeds by iteratively refining a
colouring of vertex tuples. The number of iterations needed to obtain the final output is crucial for the parallelisability of the algorithm. 

We show that there is a constant $k$ such that every planar graph can be identified (that is, distinguished from every non-isomorphic graph)
by the $k$-dimensional WL algorithm within a logarithmic number of
iterations. This generalises a result due to Verbitsky (STACS 2007), who proved the same for 3-connected planar graphs.

The number of iterations needed by the $k$-dimensional WL algorithm
to identify a graph corresponds to the quantifier
depth of a sentence that defines the graph in the $(k+1)$-variable
fragment $\LC{k+1}{}$ of first-order logic
with counting quantifiers. Thus, our result implies that every
planar graph is definable with a $\LC{k+1}{}$-sentence of
logarithmic quantifier depth.
\end{abstract}

\section{Introduction}

The Weisfeiler-Leman (WL) algorithm is a well-known combinatorial
procedure for detecting symmetries in graphs. It is widely used in approaches to tackle the graph-isomorphism problem, both from a theoretical (\cite{bab16,caifurimm92,kie20}) and from a practical perspective (\cite{darliffsakmar04,junkas07,mck81,mckaypip14}). 
The algorithm is derived from a technique called \emph{naïve vertex classification} (or \emph{Colour Refinement}), which may be viewed
as the $1$-dimensional version $\WL{1}{}$ of the WL
algorithm. For every $k\ge 1$, the $k$-dimensional WL algorithm
($\WL{k}{}$) iteratively colours $k$-tuples of vertices of a
graph by propagating local information until it reaches a \emph{stable colouring}. Weisfeiler and Leman introduced the 2-dimensional version
$\WL{2}{}$, today known as the \emph{classical WL algorithm}, in
\cite{weilem68}. The algorithm $\WL{k}{}$ can be
implemented to run in time $O(n^{k+1}\log n)$ on graphs of order $n$
\cite{immlan90}.

The algorithm has striking connections to numerous areas of mathematics
and computer science,
which surely is a reason why research on it has been active since its
introduction over half a century ago. For example, there are tight connections to
linear and semidefinite programming~\cite{atsman13,atsoch18,groott15}, homomorphism
counting~\cite{delgrorat18,dvo10}, and the algebra of coherent
configurations \cite{chepon19}. Most recently, the WL algorithm has been applied in several interesting machine-learning contexts
\cite{ahmkermlanat13,gro21,morritfey+19,sheschlee+11,xuhulesjeg19}.

A very strong and highly exploited link between the algorithm and
logic was established by Immerman and Lander \cite{immlan90} and
Cai, Fürer, and Immerman \cite{caifurimm92}: $\WL{k}{}$ assigns the same
colour to two $k$-tuples of vertices if and only if these tuples satisfy the same formulas of the $(k+1)$-variable fragment
$\LC{k+1}{}$ of first-order logic with counting
quantifiers. Cai, Fürer, and Immerman \cite{caifurimm92}
used this correspondence and an Ehrenfeucht-\Fraisse\ game
that characterises equivalence for the
logic $\LC{k+1}{}$  to prove that, for every $k$, there are
non-isomorphic graphs
of order $O(k)$ that are not distinguished by $\WL{k}{}$. Here we say that $\WL{k}{}$ \emph{distinguishes} two graphs if $\WL{k}{}$ computes different stable colourings on them, that
is, there is some colour such that the numbers of $k$-tuples of that
colour differ in the two graphs. 

We say that $\WL{k}{}$
\emph{identifies} a graph $G$ if it distinguishes $G$ from all graphs
$G'$ that are not isomorphic to $G$. It has been shown that for suitable
constants $k$, the algorithm $\WL{k}{}$ identifies all planar graphs~\cite{gro98a}, all
graphs of bounded tree width \cite{gromar99}, and all graphs in many
other natural graph classes \cite{evdpontin00,gro00,gro17,grokie19,groneu19}. For
some of these classes, fairly tight bounds for the optimal value of $k$,
called the \emph{Weisfeiler-Leman (WL) dimension}, are known. Notably,
interval graphs have WL dimension $2$ \cite{evdpontin00}, graphs of
tree width $k$ have WL dimension in the range $\lceil k/2\rceil-3$ to $k$
\cite{kieneu19}, and, most relevant for us, planar graphs have WL dimension
$2$ or $3$ \cite{kieponschwei19}. 

Another parameter of the WL
algorithm that has received recent attention is the number of
iterations it needs to reach its final, stable colouring.
Since a set of size $n^k$ can only be partitioned
$n^k-1$ times, a natural upper bound on the number of
iterations to reach the final output is $n^k-1$ ($n$ always denotes the number of
vertices of the input graph). This bound cannot be improved for $\WL{1}{}$, since there are infinitely many
graphs on which the algorithm takes $n-1$ iterations to compute its final output
\cite{kiemck20}. However, for $\WL{2}{}$, it was shown that the bound
$\Theta(n^2)$ is asymptotically not tight \cite{kieschwei19}. Currently, the best upper bound on the iteration number for $\WL{2}{}$ is
$O(n\log n)$ \cite{lichponschwei19}.

The number of iterations of $\WL{k}{}$ is crucial for the
parallelisability of the algorithm: for $\ell\ge \log n$, it holds that $\ell$
iterations of $\WL{k}{}$ can be simulated in $O(\ell)$ steps on a PRAM
with $O(n^k)$ processors \cite{grover06,KoblerV08}. In particular, if for a class
$\CC$ of graphs, all $G,G'\in\CC$ (of order $n$) can be distinguished
by $\WL{k}{}$ in $O(\log n)$ iterations, then the isomorphism problem
for graphs in $\CC$ is in the complexity class $\textsf{AC}^1$. Grohe and
Verbitsky \cite{grover06} proved that this is the case for all classes
of graphs of bounded tree width and all maps (graphs embedded into a
surface together with a rotation system specifying the embedding), and
Verbitsky~\cite{ver07} proved it for the class of 3-connected planar
graphs.

\paragraph*{Our results}

We say that $\WL{k}{}$ distinguishes two graphs \emph{in $\ell$
iterations} if the colouring obtained by $\WL{k}{}$ in the $\ell$-th
iteration differs among the two graphs, and we say $\WL{k}{}$ identifies
a graph \emph{in $\ell$ iterations} if it distinguishes the graph from
every non-isomorphic graph in $\ell$ iterations.

\begin{theorem}\label{thm:main}
There is a constant $k$ such that $\WL{k}{}$ identifies every
$n$-vertex planar graph in $O(\log n)$ iterations.
\end{theorem}

The correspondence between $\WL{k}{}$ and the logic $\LC{k+1}{}$ can
be refined to a correspondence between the number of iterations and the
quantifier depth: $\WL{k}{}$ assigns the same
colour to two $k$-tuples of vertices in the $\ell$-th iteration if and only if these two
$k$-tuples satisfy the same $\LC{k+1}{}$-formulas of quantifier depth
$\ell$. Thus, the following theorem is equivalent to Theorem~\ref{thm:main}.

\begin{theorem}\label{thm:main:logical}
  There is a constant $k$ such that for every
  $n$-vertex planar graph $G$, there is a $\LC{k}{}$-sentence of quantifier
  depth $O(\log n)$ that identifies $G$ (that is, characterises $G$ up to isomorphism).
\end{theorem}

We exploit the logical characterisation of the WL algorithm in our
proof, so it is actually Theorem~\ref{thm:main:logical} that we
prove. We first show that every planar graph $G$ has a tree
decomposition of logarithmic height where each bag consists of at most
four 3-connected components of $G$ and the adhesion is at most $6$. Then we inductively construct a formula to identify $G$ by ascending through the tree, encoding all information about isomorphism types of the parsed subgraphs in subformulas. At each node of the tree, we use
Verbitsky's result to deal with the 3-connected components.

\section{Preliminaries}\label{sec:prelim}

All graphs in this paper are finite, simple, and undirected. For a
graph $G$, we denote by $V(G)$ and $E(G)$ its set of vertices and
edges, respectively. The \emph{order} of $G$ is $|G| \coloneqq
|V(G)|$. We write edges without parenthesis, as in $vw$. For $v\in V(G)$, we let $N_G(v)\coloneqq\{w\mid vw\in
E(G)\}$. 

A \emph{subgraph} of $G$ is a graph $H$ with $V(H) \subseteq V(G)$ and
$E(H) \subseteq E(G)$. We set $N_G(H) \coloneqq \bigcup_{v \in V(H)} N_G(v) \setminus V(H)$. We call a graph $H$ a \emph{topological
  subgraph} of $G$ if a subdivision of $H$ (i.e., a graph obtained
from $H$ by replacing some edges with paths) is a subgraph of $G$.
For $W \subseteq V(G)$, we let
$G[W] \coloneqq (W, E(G) \cap \{uv\mid u,v \in W\})$ and, for arbitrary sets $W$, we let $G \setminus W\coloneqq G[V(G) \setminus W]$.

A graph $G$ is
\emph{$k$-connected} if $|G|>k$ and there is no set $S\subseteq V(G)$ with $|S| \leq k-1$ such that $G\setminus
S$ is disconnected. 

\subsection{Logic}\label{sec:log}
We denote by $\LC{}{}$ the extension of first-order logic $\FOL$ by \emph{counting
quantifiers} $\exists^{\ge m}x$ with the obvious meaning. $\LC{}{}$ is
only a syntactical extension of $\FOL$, because $\exists^{\ge
m}x\phi(x)$ is equivalent to $\exists x_1\ldots\exists
x_m\Big(\bigwedge_{i\neq j}x_i\neq
x_j\wedge\bigwedge_i\phi(x_i)\Big)$. However, we are mainly interested
in the fragments $\LC{k}{}$ of $\LC{}{}$ consisting of all formulae with at
most $k$ variables (which can, however, be reused within the formula). If $m>k$, then $\exists^{\ge m}x$ cannot be
expressed in the $k$-variable fragment of $\FOL$, this is why we add
the counting quantifiers. 

We write $\varphi(x_1,\dots,x_\ell)$ to indicate that the free
variables of $\varphi$ are among $x_1,\dots,x_\ell$. Then for a graph $G$ and
vertices $u_1,\ldots,u_\ell\in V(G)$, we write
$G\models\phi(u_1,\ldots,u_\ell)$ to denote that $G$ satisfies $\phi$ if,
for all $i$, the variable $x_i$ is interpreted by $u_i$. Moreover, we
write $\phi[G,u_1,\ldots, u_i,x_{i+1},\ldots,x_\ell]$ to denote the
set of all $(\ell-i)$-tuples $(u_{i+1},\ldots,u_\ell)$ such that
$G\models\phi(u_1,\ldots,u_\ell)$. 

The \emph{quantifier depth} $\qd(\varphi)$ of a formula $\varphi \in \LC{}{}$ is its depth of quantifier nesting. More formally,
 \begin{itemize}
   \item if $\varphi$ is atomic, then $\qd(\varphi) = 0$.
   \item $\qd(\neg \varphi) = \qd(\varphi)$.
   \item $\qd(\varphi_1 \vee \varphi_2) = \qd(\varphi_1 \land \varphi_2) = \max\{\qd(\varphi_1), \qd(\varphi_2)\}$.
   \item $\qd(\exists^{\geq p} x \varphi) = \qd(\varphi) + 1.$
 \end{itemize}
We denote the set of all $\LC{k}{}$-formulas of quantifier depth at
most $\ell$ by $\LC{k}{\ell}$.

It will often be convenient to use
asymptotic notation, such as $\LC{O(1)}{O(\log n)}$. The parameter $n$
always refers to the order of the input graph, and we will typically
make assertions such as:
{\itshape
For every $n$, there exists a $\LC{O(1)}{O(\log n)}$-formula $\phi^{(n)}(x)$
such that for all graphs $G$ of order $|G|=n$ and all $v\in V(G)$, [something holds]}. What this means is that {\itshape there is a constant
$k$ and a function $\ell(n)\in O(\log n)$ such that for every $n$, there exists a $\LC{k}{\ell(n)}$-formula $\phi^{(n)}(x)$
such that for all graphs $G$ of order $|G|=n$ and all $v \in V(G)$, [something holds].}

Throughout this paper, we will have to express properties of graphs
and their vertices using $\LC{O(1)}{O(\log n)}$-formulas. The basic
building blocks that we use are connectivity statements with formulas of logarithmic quantifier depth, as illustrated in the following example.

\begin{example}\label{exa:dist}
  For every $k\ge0$, we define a $\LC{3}{\lceil\log n\rceil}$-formula $\formel{dist}_{\leq k}$
  such that for every graph $G$ of order at most $n$ and all vertices  $u,u'\in V(G)$, it
  holds that $G \models \formel{dist}_{\le k}(u,u')$ if and only if
  $u$ and $u'$ have distance at most $k$ in $G$. We let 
  \[\formel{dist'}_{\le k}(x,x') \coloneqq
  \begin{cases}
    x = x' & \text{if } k = 0
    \\E(x,x') \lor x=x' & \text{if } k = 1
    \\
     \exists y_k \big( \formel{dist'}_{\le \lfloor \frac{k}{2}\rfloor}(x,y_k) \land \formel{dist'}_{\le \lceil \frac{k}{2}\rceil}(y_k,x')\big) & \text{otherwise.} 
  \end{cases}
  \]
  Thus, for $k \leq n$, the quantifier depth of $\formel{dist'}_{\leq k}$ is bounded by $\lceil\log n\rceil$. Now, it suffices to note that we can actually get by with the three variables $x,x', y_k$ by reusing them in the subformulas that are defined inductively. We hence obtain the desired $\LC{3}{\lceil\log n\rceil}$-formula $\formel{dist}_{\leq k}$. Note that, for $k\ge 1$, the $\LC{3}{\lceil\log n\rceil}$-formula
  $\formel{dist}_{=k}(x,x')\coloneqq\formel{dist}_{\le k}(x,x')\wedge\neg
  \formel{dist}_{\le k-1}(x,x')$ states that $x$ and $x'$ have distance
  exactly $k$. Moreover, in every graph of order at most $n$, the $\LC{3}{\lceil\log n\rceil}$-formula $\formel{comp}(x,x')\coloneqq\formel{dist}_{\le n-1}(x,x')$ states that $x$ and $x'$ lie in the same connected component and the $\LC{3}{\lceil\log n\rceil}$-sentence
  $\formel{conn}_n\coloneqq\forall x\forall x'\formel{dist}_{\le n-1}(x,x')$
  states that the graph is connected.
  \uend
\end{example}

\subsection{The WL Algorithm}
We briefly review the WL algorithm. For details, we refer to
the recent survey \cite{kie20}. 

Let $k \geq 1$. The
\emph{atomic type} $\atp(G,\bar u)$ of a $k$-tuple
$\bar u=(u_1,\ldots,u_k)$ of vertices of a graph $G$ is the set of all atomic facts satisfied by these
vertices, that is, all adjacencies and equalities between the vertices. Hence, tuples $\bar
u=(u_1,\ldots,u_k)$ and $\bar v=(v_1,\ldots,v_k)$ of vertices of
graphs $G,H$, respectively, have the same atomic type if and only if the mapping $u_i\mapsto v_i$ is an isomorphism from
the graph $G[\{u_1,\ldots,u_k\}]$ to $H[\{v_1,\ldots,v_k\}]$. 

The algorithm $\WL{k}{}$ (the \emph{$k$-dimensional Weisfeiler-Leman algorithm}) takes a graph $G$ as input and computes the following
sequence of \emph{colourings} $\logic{wl}_i^k$ of $V(G)^k$ for $i\ge0$, until it
returns $\logic{wl}_\infty^k \coloneqq \logic{wl}_i^k$ for the smallest $i$ such
that, for all $\bar u,\bar v$, it holds that $\logic{wl}_i^k(\bar u)=\logic{wl}_i^k(\bar v)\iff
\logic{wl}_{i+1}^k(\bar u)=\logic{wl}_{i+1}^k(\bar v)$. Set $\logic{wl}_0^k(\bar u) \coloneqq
\atp(G,\bar u)$.
In the $(i+1)$-st \emph{iteration}, the colouring $\logic{wl}_{i+1}^k$ is defined
by
$
\logic{wl}_{i+1}^k(\bar u) \coloneqq \big(\logic{wl}_i^k(\bar u),M_{i}(\bar u)\big),
$
where, for $\bar u=(u_1,\ldots,u_k)$, we let $M_i(\bar u)$ be the multiset
\begin{align*}
&\Bigllcurly\big(\atp(G,(u_1,\ldots,u_k,v)),\logic{wl}_i^k(u_1,\ldots,u_{k-1},v),\\
  &\hspace{3.8cm}\logic{wl}_i^k(u_1,\ldots,u_{k-2},v,u_k),\ldots,\logic{wl}_i^k(v,u_2,\ldots,u_k)\big)\mid
v\in V\Bigrrcurly
\end{align*}
The algorithm $\WL{k}{}$ \emph{distinguishes} two graphs $G$, $H$ \emph{in $\ell$ iterations} if there is
a colour $c$ in the range of $\logic{wl}_\ell^k$ such that the number of
tuples $\bar u\in V(G)^k$ with $\logic{wl}_\ell^k(\bar u)=c$ is different
from the number of
tuples $\bar v\in V(H)^k$ with $\logic{wl}_\ell^k(\bar v)=c$. In this case, we say \emph{$\WL{k}{\ell}$ distinguishes $G$ and $H$}.
Moreover, $\WL{k}{\ell}$ \emph{identifies} $G$ if it distinguishes $G$ from all graphs
$H$ that are not isomorphic to $G$. 

\begin{theorem}[\cite{caifurimm92,immlan90}]\label{thm:quantdepth}
  Let $k \in \mathbb{N}$. Let $G$ and $H$ be graphs with $|G| = |H|$ and let $\bar{u} \coloneqq (u_1,\ldots,u_k)\in V(G)^k$ and
  $\bar{v} \coloneqq (v_1,\ldots,v_k)\in V(H)^k$. Then, for all $i \in \mathbb{N}$, the following are equivalent.
    \begin{enumerate}
      \item $\logic{wl}^k_i(\bar u)=\logic{wl}^k_i(\bar v)$.
      \item
        $G\models\phi(u_1,\ldots,u_k)\iff
        H\models\phi(v_1,\ldots,v_k)$ holds for every $\LC{k+1}{i}$-formula
        $\phi(x_1,\ldots,x_k)$.
    \end{enumerate}
\end{theorem}

\section{3-Connected Planar Graphs}
\label{sec:3:connected}

Verbitsky \cite{ver07} proved that $\WL{O(1)}{O(\log n)}$
distinguishes any two 3-connected planar graphs. Before we discuss the specific version of this result
that we need here, let us briefly review some background on planar
graphs. Intuitively, a \emph{plane graph} is a graph drawn into the
plane with no edges crossing. A \emph{planar graph} is an abstract
graph $G$ isomorphic to a plane graph; an isomorphism from $G$ to
a plane graph is a \emph{planar embedding} of $G$. Now suppose $G$ is a plane graph. If we cut the plane along all edges of the graph,
the pieces that remain are the \emph{faces} of $G$ (note that one of these faces is unbounded). The closed walk
along the vertices and edges in the boundary of a face is the
\emph{facial walk} associated with this face. If $G$ is 2-connected,
then every facial walk is a cycle. If $G$ is 3-connected, we can
describe the facial cycles combinatorially: a cycle $C$ is a facial
cycle of $G$ if and only if $C$ is an induced subgraph of $G$ and
$G\setminus V(C)$ is connected. (This is the statement of
\emph{Whitney's Theorem} \cite{whi32}.) This implies that all planar
embeddings of a 3-connected planar graph have the same facial cycles,
which can be interpreted as saying that, combinatorially, all planar
embeddings of the graph are the same. Another way of describing a
planar embedding combinatorially is by specifying, for each vertex,
the cyclic order in which the edges incident to this vertex
appear. This is what is known as a \emph{rotation system}. It is easy
to see that a rotation system determines all facial walks, and,
conversely, the facial walks determine the rotation system. One last
fact that we need to know about plane graphs is \emph{Euler's formula}: if $G$ is a connected plane graph with $n$ vertices, $m$
edges, and $f$ faces, then $n-m+f=2$. (For details and more
background, we refer the reader to \cite{die16}.)

Let us now turn to the version of Verbitsky's theorem about
3-connected planar graphs that we need here. It says that, in a
3-connected planar graph, we can find three vertices such that once these
vertices are fixed, we can identify every other vertex by a
$\LC{O(1)}{O(\log n)}$-formula. 

\begin{theorem}[\cite{ver07}]\label{thm:3conn}
  Let $n \in \Nat$ and let $G$ be a 3-connected planar graph of order $|G| \leq n$ and 
  $v_1v_2\in E(G)$. Then there is a $v_3\in N_G(v_2)$ and for every
  $w\in V(G)$ a $\LC{O(1)}{O(\log n)}$-formula
  $\logic{id}_w(x_1,x_2,x_3,y)$ such that $G\models
  \logic{id}_w(v_1,v_2,v_3,w)$ and $G\not\models
  \logic{id}_w(v_1,v_2,v_3,w')$ for all $w'\in V(G)\setminus \{w\}$.
\end{theorem}

The key step in
Verbitsky's proof is to define the rotation system underlying
the unique planar embedding of a 3-connected planar graph. To state
this formally, we use the terminology of \cite{gro98a,gro17}. An
\emph{angle} of a plane graph $G$ \emph{at} a vertex $v$ is a triple $(w,v,w')$ of
vertices such that $vw$ and $vw'$ are successive edges in a facial
walk of $G$. Two angles
$(v_1,v_2,v_3)$ and $(w_1,w_2,w_3)$ are \emph{aligned}
if $w_1=v_2$ and
$w_2=v_3$ and both angles appear in the same facial walk.
Observe that, if we know all angles at a vertex $v$, we can define
the cyclic permutation of the edges incident with $v$ induced by the
embedding. If we know all angles of $G$ and the alignment relation
between them, we can define the
rotation system. 
By Whitney's Theorem, all planar embeddings of a 3-connected planar
graph $G$ have the
same angles; we call them \emph{the angles of $G$}. Similarly, we can define
abstractly if
two angles of a 3-connected planar graph are aligned.

\begin{lemma}[\cite{ver07}]\label{lem:angles}
  There are $\LC{O(1)}{O(\log n)}$-formulas
  $\logic{ang}^{(n)}(x_1,x_2,x_3)$ and
  $\logic{aln}^{(n)}(x_1,x_2,x_3,x_4)$  
  such that for all 3-connected planar graphs $G$ of order $|G|=n$ and all
  $v_1,v_2,v_3,v_4\in V(G)$, we have
  \begin{align*}
    G\models \logic{ang}^{(n)}(v_1,v_2,v_3)&\iff(v_1,v_2,v_3)\text{
      is an angle of }G,\\
G\models
    \logic{aln}^{(n)}(v_1,\ldots,v_4)&\iff(v_1,v_2,v_3), (v_2,v_3,v_4)\text{
                                                    are aligned angles of }G.
  \end{align*}
\end{lemma}

This lemma is an easy consequence of the results in
\cite[Section~4]{ver07}. The terminology there is different, the
notion corresponding to (aligned) angles is that of a \emph{layout
system}. Verbitsky's proof is based on a careful (and tedious)
analysis of how two paths between the neighbours of a vertex may
intersect.

To give the reader some intuition about the lemma, we sketch an
alternative proof, which is based on ideas from \cite{gro00} (also see
\cite[Section~10.4]{gro17}). Let $G$ be a 3-connected planar graph,
and let us think of $G$ as being embedded in the plane. It
follows from Euler's formula that in
every plane graph of minimum degree $3$, a constant fraction
of the edges is contained in facial walks
of length at most $6$. Using
Whitney's Theorem, we can define the set of all $6$-tuples that
determine a facial cycle of length at most $6$ using a $\LC{9}{}$-formula of logarithmic quantifier depth. This gives us all the angles associated with these cycles and
the alignment relation on these angles. The faces corresponding to
these facial cycles of size at most $6$ can be partitioned into \emph{regions}, where two faces belong to the same
region if their boundaries share an edge (see
Figure~\ref{fig:regions}(a)). 

We define a new graph $G^{(1)}$ as follows: for every region $R$ of
$G$, we delete all vertices contained in the interior of $R$, all
vertices on the boundary of $R$ that have no neighbours outside the region, and all edges that are either in the interior or on the
boundary of the region. Then we add a fresh vertex $v_R$
and edges from $v_R$ to all vertices that remain in the boundary of
the region $R$ (see Figure~\ref{fig:regions}(b)). Each face of
$G^{(1)}$ corresponds to a face of $G$ that we have not found
yet. Applying Euler's formula again, we can prove that a constant
fraction of the edges of $G$ that remain edges of $G^{(1)}$ are
contained in facial walks of $G^{(1)}$ that contain at most six
vertices of degree $\ge 3$. We can define the facial walks of the
corresponding edges in $G$, again using Whitney's Theorem to test if a
cycle is facial. Note that, for this, we do not need $G^{(1)}$ to be
$3$-connected (in general, it is not); we always define facial cycles in
the original graph $G$. The new facial cycles together with those
found in the first step give us new regions (covering more faces of
$G$), and from these, we construct a graph $G^{(2)}$. Iterating the
construction, we obtain a sequence of graphs $G^{(i)}$. The
construction stops once we have found all facial walks of $G$. Since
we always use a constant fraction of the edges, this happens after at
most logarithmically many iterations. This completes our proof sketch
of Lemma~\ref{lem:angles}.

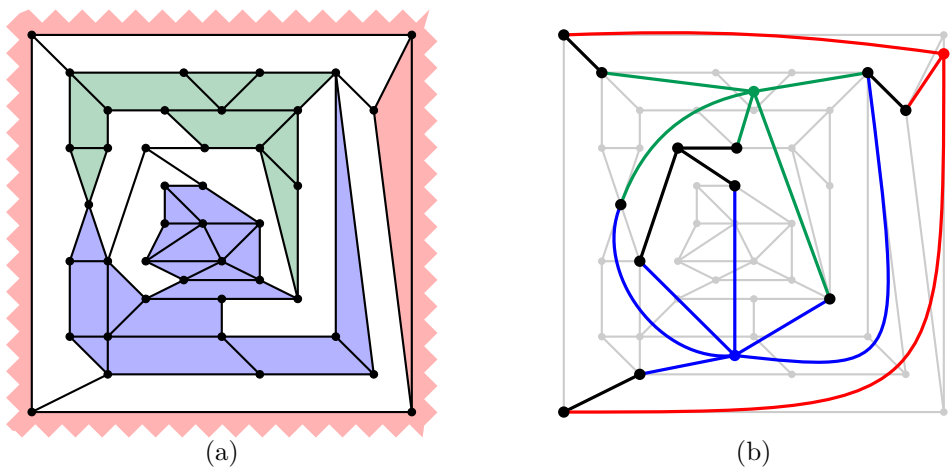
\begin{figure}
  \centering
  \begin{tikzpicture}[decoration=zigzag]

  \begin{scope}[scale=0.5]
  \fill[red!30, decorate] (-0.5,-0.5) rectangle (10.5,10.5);
  \fill[blue!30] (1,1) rectangle (9,9);
  
  \foreach \i/\x/\y in
  {1/10/10,2/10/0,3/0/0,4/0/10,5/4/9,6/6/9,7/8/9,8/9/8,9/9/1,10/6/1,11/2/1,12/1/2,13/1/4,14/1.5/5.5,15/1/7,16/1/9,17/5/8,18/7/8,19/8/2,20/5/2,21/2/2,22/2/4,23/2/7,24/2/8,25/3.5/8,26/6/7,27/7/6,28/7/3,29/5/3,30/3/3,31/3/7,32/4.55/7,33/6/5,34/6/3.5,35/4/3.5,36/3/4,37/3.5/5,38/3.5/6,39/4.5/5,40/5/4,41/4.5/6}
  \node[coordinate] (v\i) at (\x,\y) {\i};

  \fill[white] (v1) -- (v4) -- (v16) -- (v5) -- (v6) -- (v7) -- (v8)
  -- cycle;
  \fill[white] (v2) -- (v3) -- (v11) -- (v10) -- (v9) -- (v7) -- (v8)
  -- cycle;
 \fill[white] (v3) -- (v4) -- (v16) -- (v15) -- (v14) -- (v13) --
 (v12) -- (v11) -- cycle;
  \fill[white] (v22) -- (v14) -- (v23) -- (v24) -- (v25) -- (v32) --
  (v31) -- cycle;
  \fill[white] (v22) -- (v31) -- (v41) -- (v38) -- (v37) -- (v36) --
  (v35) -- (v30) -- cycle;
  \fill[white] (v31) -- (v32) -- (v26) -- (v28) -- (v34) -- (v33) --
  (v41) -- cycle;
  \fill[white] (v18) -- (v7) -- (v19) -- (v20) -- (v29) -- (v28) --
  (v27) -- cycle;
  \fill[ForestGreen!30] (v7) -- (v6) -- (v5) -- (v16) -- (v15) -- (v14) --
  (v23) -- (v24) -- (v25) -- (v32) -- (v26) -- (v28) -- (v27) -- (v18)
  -- cycle;

  \foreach \i in {1,...,41}
  \draw[fill] (v\i) circle (1mm);

  \foreach \i/\j in {1/2,2/3,3/4,4/1,5/6,6/7,7/8,8/1,8/2,7/9,9/10,10/11,3/11,11/12,12/13,13/14,14/15,15/16,4/16,16/5,17/5,17/6,17/18,18/7,7/19,19/20,19/9,20/10,20/21,21/11,21/12,21/22,22/13,22/14,23/14,23/15,23/24,24/16,24/25,25/17,26/18,26/27,26/28,27/18,27/28,28/29,29/20,29/30,30/21,30/22,31/22,31/32,32/25,32/26,31/41,41/33,41/38,33/39,33/40,33/34,34/28,34/40,34/35,35/40,35/30,35/36,36/40,36/39,36/37,37/39,37/38,38/39,39/40}
  \draw[thick] (v\i) edge (v\j);

  \path (5,-1.1) node {(a)}; 
\end{scope}


\begin{scope}[xshift=7cm,scale=0.5]
  \foreach \i/\x/\y in
  {1/10/10,2/10/0,3/0/0,4/0/10,5/4/9,6/6/9,7/8/9,8/9/8,9/9/1,10/6/1,11/2/1,12/1/2,13/1/4,14/1.5/5.5,15/1/7,16/1/9,17/5/8,18/7/8,19/8/2,20/5/2,21/2/2,22/2/4,23/2/7,24/2/8,25/3.5/8,26/6/7,27/7/6,28/7/3,29/5/3,30/3/3,31/3/7,32/4.55/7,33/6/5,34/6/3.5,35/4/3.5,36/3/4,37/3.5/5,38/3.5/6,39/4.5/5,40/5/4,41/4.5/6}
  \node[coordinate] (v\i) at (\x,\y) {\i};

  \foreach \i in {1,...,41}
  \path[fill=black!20] (v\i) circle (1mm);

  \foreach \i/\j in {1/2,2/3,3/4,4/1,5/6,6/7,7/8,8/1,8/2,7/9,9/10,10/11,3/11,11/12,12/13,13/14,14/15,15/16,4/16,16/5,17/5,17/6,17/18,18/7,7/19,19/20,19/9,20/10,20/21,21/11,21/12,21/22,22/13,22/14,23/14,23/15,23/24,24/16,24/25,25/17,26/18,26/27,26/28,27/18,27/28,28/29,29/20,29/30,30/21,30/22,31/22,31/32,32/25,32/26,31/41,41/33,41/38,33/39,33/40,33/34,34/28,34/40,34/35,35/40,35/30,35/36,36/40,36/39,36/37,37/39,37/38,38/39,39/40}
  \draw[thick,black!20] (v\i) edge (v\j);

  \foreach \i/\j in {4/16,7/8,3/11,22/31,31/41,31/32}
  \draw[very thick] (v\i) edge (v\j);
  
  \node[coordinate] (red) at (10,9.5) {};
  \fill[red] (red) circle (1.5mm);
  \draw[red,very thick] (red) edge[bend right=5] (v4) edge (v8);
  \draw[red,very thick] (red) .. controls (10,0) .. (v3);

    \node[coordinate] (blue) at (4.5,1.5) {};
  \fill[blue] (blue) circle (1.5mm);
  \draw[blue,very thick] (blue) edge[bend left=60] (v14) edge (v22) edge
  (v11) edge (v28) edge (v41);
  \draw[blue,very thick] (blue) .. controls (9,1) .. (v7);

      \node[coordinate] (ForestGreen) at (5,8.5) {};
  \fill[ForestGreen] (ForestGreen) circle (1.5mm);
  \draw[ForestGreen,very thick] (ForestGreen) edge (v16) edge (v7) edge (v28) edge
  (v32) edge[bend right] (v14);

    \foreach \i in {3,4,8,16,7,11,14,22,41,28,31,32}
  \fill (v\i) circle (1.5mm);

  \path (5,-1.1) node {(b)}; 

  \end{scope}
\end{tikzpicture}
  \caption{Defining the faces of a 3-connected planar graph: (a) shows
    a 3-connected planar graph $G$ with 3 regions formed by faces with at
  most 6 edges in their boundary; (b) shows the derived graph
  $G^{(1)}$; the faces of $G^{(1)}$ are in one-to-one correspondence
  to the white faces of $G$}
  \label{fig:regions}
\end{figure}

\begin{proof}[Proof of Theorem~\ref{thm:3conn}]
  Let $G$ be a 3-connected planar graph of order $|G|=n$. For angles
  $\vec v=(v_1,v_2,v_3)$, $\vec w=(w_1,w_2,w_3)$, we write
  $\vec v\curvearrowright\vec w$ if $\vec v,\vec w$ are aligned, and
  we write $\vec v\curlywedgedownarrow\vec w$ if $w_1=v_3$ and
  $w_2=v_2$ and $w_3\neq v_1$. Note that, for every angle $\vec v$,
  there is a unique $\vec w$ such that
  $\vec v\curvearrowright\vec w$, because, by the $3$-connectedness of
  $G$, every angle is in the boundary of a unique face, and the aligned
  angle belongs to the same face. There is also a unique $\vec w'$ such that
  $\vec v\curlywedgedownarrow\vec w'$, determined by the cyclic order of
  the edges and faces around a vertex. An \emph{angle walk} is a
  sequence $\vec v_0,\ldots,\vec v_\ell$ of angles such that for all
  $i\in[\ell]$, we have $\vec v_{i-1}\curvearrowright\vec v_i$ or
  $\vec v_{i-1}\curlywedgedownarrow \vec v_i$. The \emph{direction} of
  the angle walk $\vec v_0,\ldots,\vec v_\ell$ is the tuple
  $\vec\delta=(\delta_1,\ldots,\delta_\ell)\in\{\curvearrowright,
  \curlywedgedownarrow\}^\ell$ such that for every $i\in[\ell]$, we
  have $\vec v_{i-1}\delta_i\vec v_i$. Using Lemma \ref{lem:angles}, it is straightforward to prove that for every
  $\vec\delta\in\{\curvearrowright, \curlywedgedownarrow\}^{\le n}$, there
  is a $\LC{O(1)}{O(\log n)}$-formula $\logic{awalk}^{(n)}_{\vec \delta}(\vec x,\vec y)$
  such that for all
  $\vec v,\vec w\in V(G)^3$, we have
  $
    G\models \logic{awalk}^{(n)}_{\vec
      \delta}(\vec v,\vec w)$ if and only if there is an angle walk
      of direction $\vec\delta$ from $\vec v$ to $\vec w$.
  Now let $v_1v_2\in E(G)$. Then there is a $v_3$ such that
  $(v_1,v_2,v_3)$ is an angle. Let $\vec v \coloneqq (v_1,v_2,v_3)$. Note that,
  for every $w\in V(G)\setminus\{v_1,v_2,v_3\}$, there is an angle walk
  of length at most $n$ from $\vec v$ to some $\vec w=(w_1,w_2,w_3)$
  with $w_3=w$, simply because every path in $G$ can be extended to an
  angle walk. Let $\Delta(w)$ be the set of all directions
  $\vec\delta$ of length at most $n$ such that there is an angle walk
  of direction $\bar\delta$ from $\vec v$ to some
  $\vec w=(w_1,w_2,w_3)$ with $w_3=w$. Note that the sets $\Delta(w)$
  for $w\in V(G)\setminus\{v_1,v_2,v_3\}$ are mutually disjoint.  Let
  $\logic{id}_{\bar\delta}(x_1,x_2,x_3,y)\coloneqq \exists y_1\exists
  y_2\logic{awalk}^{(n)}_{\vec \delta}(x_1,x_2,x_3,y_1,y_2,y)$. Then
  for $\bar\delta\in\Delta(w)$, we have
  $G\models \logic{id}_{\bar\delta}(v_1,v_2,v_3,w)$ and
  $G\not\models \logic{id}_{\bar\delta}(v_1,v_2,v_3,w')$ for all
  $w'\neq w$. 
\end{proof}

\section{Decomposition into Blocks}\label{sec:tree:decompositions}

Let $G$ be a graph. A \emph{tree decomposition} of $G$ is a pair $(T,\beta)$ where $T$ is a tree and $\beta\colon V(T) \rightarrow 2^{V(G)}$ is a function such that for every $v \in V(G)$, the set $\{t \in V(T) \mid v \in \beta(t)\}$ is non-empty and induces a connected subgraph in $T$, and
for every $e \in E(G)$, there is a $t \in V(T)$ such that $e \subseteq \beta(t)$.
For $t \in V(T)$, we call $\beta(t)$ a \emph{bag} of $(T,\beta)$. 
The \emph{adhesion} of $(T,\beta)$ is
$
  \ad(T,\beta) \coloneqq \max\big\{|\beta(t)\cap\beta(u)|\bigmid tu\in
  E(T)\big\}
$
(or $0$ if $E(T)=\emptyset$). The \emph{width} of $(T,\beta)$ is
$\width(T,\beta) \coloneqq \max_{t \in V(T)} |\beta(t)| - 1.$

We denote the root of a rooted tree $T$ by $r^T$. For better readability, if the rooted tree is referred to as $T^*$, we set $r^* \coloneqq r^{T^*}$. The \emph{height} of
$T$ is the maximum length of a path from $r^T$ to a leaf of $T$. We
denote the descendant order of $T$ by $\trianglelefteq^T$. That is,
$t\trianglelefteq^T u$ if $t$ occurs on the path from $r^T$ to $u$.
A \emph{rooted tree decomposition} is a tree decomposition where the
tree is rooted. 

\begin{lemma}[Folklore]\label{lem:1}
  Let $T$ be a tree and $\chi \colon V(T)\to\Real_{\ge0}$. Then there is a node
  $t\in V(T)$ such that for every connected component $C$ of
  $T\setminus \{t\}$, it holds that
  \[
    \sum_{t\in V(C)}\chi(t)\le\frac{1}{2}\sum_{t\in V(T)} \chi(t).
  \]
\end{lemma}

\begin{proof}
  Orient all edges towards the larger sum of $\chi$-weights in the connected components that the removal of the edge would induce, breaking
  ties arbitrarily. There will
  be a node such that all incident edges are oriented towards it. This node has the desired property.
\end{proof}

The following lemma is known in its essence (for example,
\cite{elbjaktan10}), though we are not aware of a reference where it
is stated in this precise form, which we will need later.

\begin{lemma}\label{lem:logdec}
  Let $T$ be a tree, and let $B\subseteq
  V(T)$ be a set of size $|B|\le 3$. Then there is a rooted tree
  decomposition $(T^*,\beta^*)$ of $T$ with $B \subseteq \beta^*(r^*)$ and the following additional properties.
  \begin{eroman}
  \item\label{it:ld1} The height of $T^*$ is at most $2\log |T|$. 
  \item\label{it:ld2} The width of $(T^*,\beta^*)$ is at most $3$.
  \item\label{it:ld3} The adhesion of $(T^*,\beta^*)$ is at most $3$.
  \item\label{it:ld5}
    For every $t^*\mspace{-4mu}\in\mspace{-3mu} V(T^*)$ and every child $u^*\mspace{-4mu}$ of $t^*\mspace{-4mu}$, the graph
    $T\mspace{-2mu}\left[(\bigcup_{v^*\trianglerighteq^{T^*}
        u^*}\beta^*\mspace{-2mu}(v^*))\setminus\beta^*(t^*)\right]$ is
    connected.
  \end{eroman}
\end{lemma}

\begin{proof}
  Condition \eqref{it:ld5} is something that we can easily achieve for every
  rooted tree decomposition: if, for the rooted subtree at some node, the subgraph induced by the bags in this subtree is not connected, we simply create one copy of the subtree for each connected component and only keep the vertices of that connected
  component in the copy. Moreover, the adhesion of a tree decomposition of width $3$ can only
  be larger than 3 if there are adjacent nodes with the same
  bag. If this is the case, we can simply contract the edge between
  the nodes. Repeating this, we can turn the decomposition into a
  decomposition of adhesion at most $3$. So we only need to take care
  of Conditions \eqref{it:ld1} and \eqref{it:ld2}.

  The proof is by induction on $n\coloneqq|T|$. We prove a slightly
  stronger statement; in addition to $B\subseteq\beta^*(r^*)$, we
  require $|\beta^*(r^*)\setminus B|\le 1$.

  The base case $n\le 4$ is easy: for $n=1$, the $1$-node tree
  decomposition of height $0$ has all the desired properties, and for
  $2\le n\le 4$, we can take a $2$-node tree decomposition of height
  $1$ where the root bag is $B$ and the leaf bag is $V(T)$.

  For the inductive step, suppose $n> 4$.
  \begin{cs}
    \case1
    $|B|<3$.\\
    By Lemma~\ref{lem:1}, there is a node $b\in V(T)$ such that for
    every connected component $C$ of $T\setminus\{b\}$, it holds that
    \[
      |V(C)|\le\frac{n}{2}.
    \]
    Let $C_1,\ldots,C_m$ be the vertex sets of the connected
    components of $T\setminus\{b\}$.
    For every $i\in[m]$, let $c_i$ be the unique neighbour
    of $b$ in $C_i$, and let $B_i\coloneqq (B\cap
    V(C_i))\cup\{c_i\}$. Note that $|B_i|\le 3$.
    
    By the induction hypotheses, for every $i$, there is a rooted tree
    decomposition $(T_i,\beta_i)$ of $C_i$ with the desired
    properties. In particular, the height of $T_i$ is at most
    $2\log(n/2)=2\log n-2$.

    For every $i$, let $r_i$ be the root of $T_i$. We form a new tree
    $T^*$ by taking the disjoint union of all the $T_i$, adding fresh
    nodes $r^*$ and $r_i^*$ for $i\le m$, and adding edges $r^*r_i^*$,
    $r_i^*r_i$ for all $i\in[m]$. We define
    $\beta^* \colon V(T^*)\to 2^{V(T)}$ by
    \[
      \beta^*(t)\coloneqq
      \begin{cases}
        B\cup\{b\}&\text{if }t=r^*,\\
        B_i\cup\{b\}&\text{if }t=r_i^*,\\
        \beta_i(t)&\text{if }t\in V(T_i).
      \end{cases}
    \]
    Then $(T^*,\beta^*)$ is a tree decomposition of $T$ of width at
    most $3$ and height at most $2\log n$.
    
    \case2 $|B|=3$.\\
    By Lemma~\ref{lem:1} applied to the characteristic function of
    $B$, there is a node $b\in V(T)$ such that for every connected
    component $C$ of $T\setminus\{b\}$, it holds that
    \[
      |V(C)\cap B|\le1.
    \]
    Let $C_1,\ldots,C_\ell$ be the connected components of
    $T\setminus\{b\}$, and for every $i$, let $B_i\coloneqq B\cap
    V(C_i)$. Then $|B_i|\le 1$.
    
    \begin{claim}\label{cl:logdec1}
      For every $i\in[\ell]$, there is a tree decomposition
      $(T_i,\beta_i)$ of width at most $3$ such that the height of
      $T_i$ is at most $2\log n-1$ and for the root $r_i$ of $T_i$ it
      holds that $B_i\subseteq\bag_i(r_i)$ and
      $|\bag(r_i)|\le 2$.
    \end{claim}

    \begin{claimproof}
      Let $\in[\ell]$
    and $n_i\coloneqq|C_i|$. By Lemma~\ref{lem:1}, there is a $c\in V(C_i)$ such
    that for every connected component $D$ of $C_i\setminus\{c\}$, it
    holds that $|D|\le n_i/2$. Choose such a $c$ and let
    $D_{1},\ldots,D_{m}$ be the connected components of
    $C_i\setminus\{c\}$. For every $j\in[m]$, let $d_{j}$ be the
    unique neighbour of $c$ in $D_j$. Let
    $B_{ij}\coloneqq (B_i\cap D_j)\cup\{d_j\}$. Then $|B_{ij}|\le 2$.

    By the induction hypotheses, for every $j$, there is a rooted
    tree decomposition $(T_{ij},\beta_{ij})$ of $D_j$ of width $3$
    such that the height of $T_{ij}$ is at most
    $2\log|D_i|\le2\log(n_i/2)\le 2\log n-2$. Furthermore, for the
    root $r_{ij}$ of $T_{ij}$, it holds that
    $B_{ij}\subseteq\beta_{ij}(r_{ij})$ and
    $|\beta_{ij}(r_{ij})\setminus B_{ij}|\le 1$. This implies
    $|\beta_{ij}(r_{ij})|\le 3$.

    We form a new tree $T_i$ by taking the disjoint union of all the
    $T_{ij}$ for $j\in[m]$, adding a fresh node $r_i$, and adding
    edges $r_ir_{ij}$ for all $j\in[m]$. We define
    $\beta_i \colon V(T_i)\to 2^{V(C_i)}$ by
    \[
      \beta_i(t)\coloneqq
      \begin{cases}
        B_i\cup\{c\}&\text{if }t=r_i,\\
        \beta_{ij}(r_{ij})\cup\{c\}&\text{if }t=r_{ij},\\
        \beta_{ij}(t)&\text{if }t\in V(T_{ij})\setminus\{r_{ij}\}.
      \end{cases}
    \]
    Then $(T_i,\beta_i)$ is a tree decomposition of $C_i$ with the
    desired properties.
    \end{claimproof}

    To complete the proof of the lemma, we form a new tree $T^*$ by taking the disjoint union of the
    $T_{i}$ of Claim~\ref{cl:logdec1} for $i\in[\ell]$, adding a fresh node $r^*$, and adding
    edges $r^*r_{i}$ for all $i\in[\ell]$. We define
    $\beta^* \colon V(T^*)\to 2^{V(T)}$ by
    \[
      \beta^*(t)\coloneqq
      \begin{cases}
        B\cup\{b\}&\text{if }t=r^*,\\
        \beta(r_i)\cup\{b\}&\text{if }t=r_{i},\\
        \beta_{i}(t)&\text{if }t\in V(T_{i})\setminus\{r_{i}\}.
      \end{cases}
    \]
    Then $(T^*,\beta^*)$ is a tree decomposition of $T$ of width at
    most $3$ and height at most $2\log n$.
    \qedhere
  \end{cs}
\end{proof}

Let us now turn to decompositions of a graph into its 3-connected
components. We need a few more definitions.  In the following, let $G$
be a connected graph and $X\subseteq V(G)$. The \emph{torso} of $X$ is the
graph $G\llbracket X\rrbracket$ with vertex set $X$ and edge set
\[
\Big\{vw\in\binom{X}{2}\Bigmid vw\in E(G)\text{ or
}v,w\in N_G(C)\text{ for some connected component $C$ of }G\setminus X\Big\}.
\]
The
\emph{adhesion} of $X$ is the maximum of $|N_G(C)|$ for all connected
components $C$ of $G\setminus X$. It is easy to see that if the
adhesion of $X$ is at most $2$, then the torso $G\llbracket
X\rrbracket$ is a topological subgraph of $G$ and if the
adhesion of $X$ is at most $1$, then the torso $G\llbracket
X\rrbracket$ is just the induced subgraph $G[X]$.

A \emph{block}\footnote{Our usage of the term ``block'' is non-standard. If anything, what we call a ``block'' might better be called ``2-block''. But just using ``block'' is more convenient.}
of $G$ is a set $B\subseteq V(G)$ such that
\begin{itemize}
\item either $G\llbracket B\rrbracket$ is 3-connected and the adhesion
of $B$ is at most $2$,
\item or $G\llbracket B\rrbracket$ is a complete graph of order $3$ and the adhesion
of $B$ is at most $2$, 
\item or $G\llbracket B\rrbracket$ is a complete graph of order $2$ and the adhesion
of $B$ is at most $1$. 
\end{itemize}
We call blocks with $3$-connected torsos \emph{proper blocks} and blocks
of cardinality at most $3$ \emph{degenerate blocks} of \emph{order}
3 and 2, respectively.  It is easy
to see that for distinct blocks $B,B'$, neither $B\subseteq B'$ nor
$B'\subseteq B$ holds and, furthermore, $|B\cap B'|\le 2$.
A \emph{block separator} is a set
$S\subseteq V(G)$ such that there are distinct blocks $B,B'$ with
$S=B\cap B'$, and the two sets $B\setminus S$ and $B'\setminus S$
belong to different connected components of $G\setminus S$. Note that by the definition of blocks, block separators have cardinality at most $2$.

Observe that the torsos of all blocks of a graph are topological
subgraphs. As all topological subgraphs of a planar graph are planar,
the torsos of the blocks of a planar graph are planar. In particular,
the torsos of proper blocks are 3-connected planar graphs. This will
be important later.

Call a tree decomposition $(T,\beta)$ \emph{small} if for all distinct nodes $t,u\in V(T)$, it holds that $\beta(t)\not\subseteq\beta(u)$.

\begin{lemma}[\cite{tut84}]\label{lem:3cc}
Every connected graph $G$ has a small tree decomposition $(T,\beta)$ of
adhesion at most $2$ such that for all $t\in V(T)$, the bag $\beta(t)$ is a block of $G$.
\end{lemma}

The decomposition in this lemma is essentially Tutte's well-known
decomposition of a graph into its 3-connected components described in
a slightly non-standard way. The two main differences are that,
normally, the decomposition is only described for 2-connected graphs,
whereas arbitrary connected graphs are first decomposed into their
2-connected components. We merge these decompositions into one. The
second difference is that Tutte decomposes a 2-connected graph into
3-connected pieces (our proper blocks) and cycles. Instead of cycles,
we only allow triangles, i.e., degenerate blocks of order $3$. This is
possible because every cycle can be decomposed into triangles. What we
lose with our form of decomposition is the canonicity: a graph may
have several structurally different decompositions of the form described
in the lemma. 

In the following, we apply Lemma~\ref{lem:logdec} to the tree
of the decomposition of Lemma~\ref{lem:3cc} and obtain a decomposition
of logarithmic height that is still essentially a decomposition into
3-connected components.

\begin{lemma}\label{lem:tree:dec:log:height}
Every connected graph $G$ has a rooted tree decomposition $(T^*,\beta^*)$ with the
following properties.
\begin{eroman}
\item\label{it:tdl1} The height of $T^*$ is at most $2\log|G|$. 
\item\label{it:tdl2} For every $t^*\in V(T^*)$, there are sets $B_1,\ldots,B_4$ (not
 necessarily distinct or disjoint) such that $\beta^*(t^*)=\bigcup_{i=1}^4B_i$ and each $B_i$ is either a block or a block separator.
\item\label{it:tdl3} The adhesion of $(T^*,\beta^*)$ is at most $6$.
\item\label{it:tdl6} For every $t^*\in V(T^*)$ and every child $u^*$ of $t^*$, the induced subgraph
\[
G\left[\Big(\bigcup_{v^*\trianglerighteq^{T^*}u^*}
\beta^*(v^*)\Big)\setminus\beta^*(t^*)\right]
\]
    is connected.
  \end{eroman}
\end{lemma}

\begin{proof}
  Let $(T,\beta)$ be the decomposition of $G$ into its blocks obtained
  from Lemma~\ref{lem:3cc}.  Let $(T^*,\beta_T^*)$ be the rooted tree
  decomposition of $T$ obtained from Lemma~\ref{lem:logdec}. Let $r^*$
  be the root of $T^*$, and let
  $\trianglelefteq^*{\coloneqq}\trianglelefteq^{T^*}$ be the partial
  descendant 
  order associated with $T^*$. For every
  $t^*\in V(T^*)$, let
  \begin{align*}
    \gamma^*_T(t^*)&{\coloneqq}\bigcup_{u^*\trianglerighteq^*t^*}\beta_T^*(u^*)\\
  \intertext{and}
    \sigma^*_T(t^*)&{\coloneqq}
                     \begin{cases}
                       \emptyset&\text{if }t^*=r^*,\\
                       \beta^*_T(s^*)\cap\beta^*_T(t^*)&\text{for the parent $s^*$ of $t^*$ in $T^*$, otherwise} \, .
                     \end{cases}
  \end{align*}
    For every $t\in V(T)$, we let $\min^*(t)$ be the
  unique $\trianglelefteq^*$-minimal node $t^*\in V(T^*)$ such that
  $t\in\beta_T^*(t^*)$. The uniqueness follows from the fact that the
  set of all $t^*\in V(T^*)$ with $t\in\beta_T^*(t^*)$ is connected in
  $T^*$.

  Let us call $t\in V(T)$ \emph{active} in $t^*\in V(T^*)$ if
  $t\in\bag^*_T(t^*)$ and 
  $t^*\neq\min^*(t)$ and there is a $u\in N_T(t)$ such that
  $t^*\trianglelefteq\min^*(u)$. We call $u$ an \emph{activator} of
  $t$ in $t^*$.

  \begin{claim}\label{cl:intersection:1}
    Suppose that $t\in V(T)$ is active in $t^*\in V(T^*)$. Then there
    is a unique activator of $t$ in $t^*$.
  \end{claim}
  \begin{claimproof}
    Since $t\in\bag^*_T(t^*)$ and $t^*\neq\min^*(t)$, we have
    $\min^*(t)\triangleleft t^*$ and
    $t\in\bag^*_T(\min^*(t))\cap\bag^*_T(t^*)\subseteq\sigma^*_T(t^*)$. Moreover,
    for every activator $u$ of $t$, it holds that
    $t^*\trianglelefteq\min^*(u)$, which implies
    $u\in\gamma^*_T(t^*)\setminus\sigma^*_T(t^*)$.

    Suppose towards a contradiction that $t$ has two activators $u_1,u_2$ in
    $t^*$. Then $u_1,u_2\in
    N_T(t)\cap\big(\gamma^*_T(t^*)\setminus\sigma^*_T(t^*)\big)$.
    By Lemma~\ref{lem:logdec}\eqref{it:ld5}, the induced subgraph
    $T[\gamma^*_T(t^*)\setminus\sigma^*_T(t^*)]$ is connected.
     Thus,
    there is a path from $u_1$ to $u_2$ in
    $T[\gamma^*_T(t^*)\setminus\sigma^*_T(t^*)]$. As $u_1,u_2\in
    N_T(t)$ and $t\in\sigma^*_T(t^*)$, there is a cyle in $T$, which is a contradiction.
  \end{claimproof}

Hence, in the following we can speak of \emph{the} activator of a node.
  Observe that if $t$ is active in $t^*$, then $t$ is also active in
  all $u^*$ with $\min^*(t)\triangleleft u^*\triangleleft t^*$, with
  the same activator.

  Now we are ready to define our tree decomposition $(T^*,\beta^*)$ of
  $G$. The tree $T^*$ is the same as in the decomposition
  $(T^*,\beta_T^*)$ of $T$. We
  define $\beta^* \colon V(T^*)\to 2^{V(G)}$ by letting
   $\beta^*(t^*)$ for $t^*\in V(T^*)$
  be the union of the following
  sets:
  \begin{itemize}
  \item for all $t\in\beta_T^*(t^*)$ such that
    $t^*=\min^*(t)$: the block $\beta(t)$, and
  \item for all $t\in\beta_T^*(t^*)$ such that
    $t$ is active in $t^*$ with activator $u$: the block separator
    $\beta(t)\cap\beta(u)$.    
  \end{itemize}

  \begin{claim}
    $(T^*,\beta^*)$ is a tree decomposition of $G$.
  \end{claim}

  \begin{claimproof}
    Every edge $e\in E(G)$ is contained in some bag $\beta(t)$, and
    $\beta(t)\subseteq \beta^*(\min^*(t))$.

    Now consider a vertex $v\in V(G)$. Let
    \begin{align*}
      S_v&{{}\coloneqq{}}\{t\in V(T)\mid v\in\beta(t)\},\\
      S^*_v&{{}\coloneqq{}}\{t^*\in V(T^*)\mid S_v\cap\beta^*_T(t^*)\neq\emptyset\}.
    \end{align*}
    Since $(T,\beta)$ is a tree decomposition, $S_v$ is connected in
    $T$, and as $(T^*,\beta_T^*)$ is a tree decomposition, $S_v^*$ is
    connected in $T^*$. Thus, there is a unique
    $\trianglelefteq^*$-minimal node $s^*$ in $S_v^*$. Let $s\in
    S_v\cap\beta^*_T(s^*)$. Then
    $s^*=\min^*(s)$ and therefore $v\in \beta^*(s^*)$.

    Let $t^*\in V(T^*)$ such that $v\in\beta^*(t^*)$.
    We shall prove that $v\in\beta^*(v^*)$ for all $v^*$ on the path
    from $t^*$ to $s^*$. This will prove that the set of all $t^*$
    for which $v\in\beta^*(t^*)$ holds is connected in $T^*$.

    By the definition of $\beta^*$, since $v\in\beta^*(t^*)$, there is
    a $t\in \beta^*_T(t)$ such that $v\in\beta(t)$ and either
    $t^*=\min^*(t)$ or $t$ is active in $t^*$. We choose such a
    $t$. Then $t\in S_v$ and therefore $t^*\in S_v^*$. By the
    minimality of $s^*$, this implies $s^*\trianglelefteq^* t^*$.

    The proof that $v\in\beta^*(v^*)$ holds for all $v^*$ on the path from
    $t^*$ to $s^*$ is by induction on the distance $d$ between $t^*$
    and $s^*$. The base case $d=0$ is trivial. So let us assume that
    $d\ge 1$. It follows from the definition of $\beta^*$ that
    $v\in\beta^*(v^*)$ holds for all $v^*$ on the path from $t^*$ to
    $\min^*(t)$.
    Thus, without loss of generality, we may
    assume that $t^*=\min^*(t)$.

    Let $t=t_1,\ldots,t_m=s$ be the path from $t$ to $s$ in $T$. Note
    that $v\in\beta(t_i)$ holds for all $i\in[m]$. The edge $tt_2=t_1t_2$
    must be covered by some bag $\beta_T^*(u^*)$ that contains both
    $t$ and $t_2$. Since $t^*=\min^*(t)$, we have
    $t^*\trianglelefteq^* u^*$. As the pre-image of the path
    $t_1,\ldots,t_m$ in $T^*$ is connected and
    $s^*\trianglelefteq^* t^*\trianglelefteq^* u^*$, there is an $i>1$
    such that $t_i\in \bag^*(t^*)$. If $\min^*(t_i)=t^*$, we find a
    $j>i$ such that $t_j\in \bag^*(t)$, and, repeating this, we
    eventually arrive at a $t_k\in \bag^*(t)$ such that
    $\min^*(t_k)\triangleleft t^*$. Arguing as above, we find that
    $v\in\bag^*(v^*)$ holds for all $v^*$ on the path from $t^*$ to
    $\min^*(t_k)$. Since $\min^*(t_k)$ is closer to $s^*$ than $t^*$,
    we can now apply the induction hypothesis to conclude that
    $v\in\bag^*(v^*)$ holds for all $v^*$ on the path from $\min^*(t_k)$ to
    $s^*$.
  \end{claimproof}

  Let us turn to proving that the tree decomposition $(T^*,\beta^*)$ has
  the desired properties. 

  Since $(T,\beta)$ is a small decomposition, we have $|T|\le|G|$. Thus, Condition \eqref{it:tdl1} follows from Lemma~\ref{lem:logdec}\eqref{it:ld1}.

  Condition \eqref{it:tdl2} follows immediately from Lemma~\ref{lem:logdec}\eqref{it:ld2} and the definition of $\beta^*(t)$.

  To prove Condition \eqref{it:tdl3}, let $u^*$ be a child of
  $t^*$. Let us assume that $\bag^*_T(t^*)=\{t_1,\ldots,t_4\}$ and
  $\bag^*_T(u^*)=\{u_1,\ldots,u_4\}$ with $t_1=u_1,t_2=u_2$, and
  $t_3=u_3$ and $t_4\neq u_i,u_4\neq t_i$ for $i\in[4]$. The cases of
  smaller bags $\beta^*_T(t^*)$, $\beta^*_T(u^*)$ or a smaller
  intersection between them can be dealt with similarly.

  Let us first deal with the common elements $t_i=u_i$ for
  $i\in[3]$. Note that $\min^*(t_i)\trianglelefteq t^*\triangleleft
  u^*$. If $t_i$ is not active in $u^*$, then it does not
  contribute to the $\beta^*(u^*)$ and hence not to the intersection
  of the two bags. 
  If $t_i$ is active in $u^*$, say, with activator
  $v_i$, then the block separator $S_i{{}\coloneqq{}}\beta(t_i)\cap\beta(v_i)$ is
  contained in $\beta^*(u^*)$. To simplify the notation, in the
  following, we let $S_i \coloneqq 0$ if $t_i$ is not active in $u^*$.

  Either $t_i$ is active in
  $t^*$ as well with the same activator and we have
  $S_i\subseteq\beta^*(t^*)$, or
  $t^*=\min^*(t_i)$ and $S_i\subseteq\beta(t_i)\subseteq\bag^*(t^*)$. In both
  cases, 
  \begin{equation}
    \label{eq:1}
    S_i\subseteq\beta^*(t^*)\cap\beta(u^*).
  \end{equation}
  Next, let us look at the contribution of $t_4$ and $u_4$. The
  contribution of $t_4$ to $\bag^*(t^*)$ is contained in $\bag(t_4)$, and
  the contribution of $u_4$ to $\bag^*(u^*)$ is contained in $\bag(u_4)$.
  Since the only neighbour of $t_i$ in
  $\gamma_T^*(u^*)\setminus\sigma^*_T(u^*)=\gamma_T^*(u^*)\setminus\{t_1,t_2,t_3\}$
  is $v_i$ (if $t_i$ is active in $u^*$, otherwise there is no
  neighbour), all paths from $t_i$ to $u_4$ go through $v_i$. This
  implies that
  \begin{equation}
    \label{eq:2}
    \bag(t_i)\cap\bag(u_4)\subseteq \bag(t_i)\cap\bag(v_i)=S_i.
  \end{equation}
  All paths from $t_4$ to $u_4$ go through $t_1,t_2,t_3$, and
  therefore
  \begin{equation}
    \label{eq:3}
    \bag(t_4)\cap \bag(u_4)\subseteq\bigcup_{i=1}^3
    \bag(t_i)\cap\bag(u_4)\subseteq S_1\cup S_2\cup S_3.
  \end{equation}
  Thus, overall, we have $\bag^*(t^*)\cap\bag^*(u^*)\subseteq S_1\cup
  S_2\cup S_3$.

  To prove that Condition \eqref{it:tdl6} holds, let $t^*\in V(T^*)$ and and let $u^*$ be a child of
  $t^*$. To simplify the notation, let
  $\sigma^*(u^*){{}\coloneqq{}}\beta(u^*)\cap\beta^*(t^*)$ and 
  \begin{equation}
    \label{eq:4}
     \gamma^*(u^*){{}\coloneqq{}}\bigcup_{v^*\trianglerighteq
      u^*}\beta^*(v^*).
    \end{equation}
  We need to prove that $G[\gamma^*(u^*)\setminus\sigma^*(u^*)]$ is
  connected. The key observation is that
 \begin{equation}
    \label{eq:5}
    \gamma^*(u^*)\setminus\sigma^*(u^*)=\bigcup_{t\in\gamma_T^*(u^*)\setminus
      \sigma_T^*(u^*)}\bag(t).
  \end{equation}
  The reason for this is that, for all $t\in\gamma_T^*(u^*)\setminus
      \sigma_T^*(u^*)$, it holds that $u^*\trianglelefteq\min^*(t)$,
      which implies that $\beta(t)\subseteq\beta^*(\min^*(t))$ appears
      on the right-hand side of \eqref{eq:4}. It follows from Part~\eqref{it:ld5} in Lemma~\ref{lem:logdec} that the set $\gamma_T^*(u^*)\setminus \sigma_T^*(u^*)$ is connected in $T$, and this implies that the union on the right-hand side of \eqref{eq:5} is connected.
\end{proof}

Our next goal will be to define the decomposition in the logic
$\LC{O(1)}{O(\log n)}$. The following lemma yields a way to define blocks via triplets of vertices.

\begin{lemma}\label{lem:block:characterisation}
Let $G$ be a graph, and let $B$ be a proper block of $G$. Let
$b_1,b_2,b_3\in B$ be pairwise distinct vertices. Then $B$ is the
set of all $v\in V(G)$ such that there is no set
$S\subseteq V(G)\setminus\{v\}$ of cardinality at most $2$
separating $v$ from $\{b_1,b_2,b_3\}$.
\end{lemma}

\begin{proof}
Let $v\in B$.
Since $G\llbracket B\rrbracket$ is 3-connected, there are paths
$P_i\subseteq G\llbracket B\rrbracket$ from $v$ to $b_i$ that are
internally disjoint, that is, $V(P_i)\cap V(P_j)=\{v\}$ for
$i\neq j$. As $G\llbracket B\rrbracket$ is a topological subgraph of
$G$, these paths can be expanded to paths $P_i'$ from $v$ to $b_i$
in $G$, and the $P_i'$ are still internally disjoint. Since every
$S\subseteq V(G)\setminus\{v\}$ of cardinality at most $2$
has an empty intersection with at least one of the paths $P_i'$, it does not separate $v$
from $\{b_1,b_2,b_3\}$.

Conversely, let $v\in V(G)\setminus B$, and let $C$ be the connected
component of $G\setminus B$ with $v\in V(C)$, and let
$S\coloneqq N_G(C)$. Then $|S|\le 2$. Then $S$ separates $v$ from
$\{b_1,b_2,b_2\}$.
\end{proof}

Let $G$ be a graph and $S,X\subseteq V(G)$. We say that $S$
\emph{separates} $X$ if there are two distinct connected components
$C_1,C_2$ of $G\setminus S$ such that $X\cap V(C_i)\neq\emptyset$ for both $i=1,2$.

\begin{lemma}\label{lem:block2}
Let $G$ be a graph, and let $b_1,b_2,b_3\in V(G)$ be mutually
distinct. Then there is a proper block $B$ with $b_1,b_2,b_3\in B$
if and only if there is a vertex $b_4\in V(G)\setminus\{b_1,b_2,b_3\}$ such that no
set $S\subseteq V(G)$ of cardinality at most $2$ separates
$\{b_1,b_2,b_3,b_4\}$.
\end{lemma}

\begin{proof}
For the forward direction, suppose that $B$ is a proper block with
$b_1,b_2,b_3\in B$. Let $b_4\in B\setminus\{b_1,b_2,b_3\}$. Then it
follows from Lemma~\ref{lem:block:characterisation} that there is
no $S$ of cardinality at most $2$ that separates $\{b_1,b_2,b_3,b_4\}$.

For the backward direction, let $B$ be the set of all $v\in V(G)$
such that no set $S\subseteq V(G)\setminus\{v\}$ of cardinality at
most $2$ separates $v$ from $\{b_1,b_2,b_3\}$. Then
$b_1,b_2,b_3\in B$ and $|B|\ge 4$. It is easy to prove that $B$ is a block.
\end{proof}

\begin{lemma}\label{lem:block}
For all $n \in \Nat$, there exist $\LC{O(1)}{O(\log n)}$-formulas
$\logic{block}^{(n)}(x_1,x_2,x_3,y)$ and $\logic{torso}^{(n)}(x_1,x_2,x_3,y,z)$ such that for all graphs $G$ of order at most $n$ and all $b_1,b_2,b_3,v\in V(G)$, we have
\[
G\models \logic{block}^{(n)}(b_1,b_2,b_3,v)
\]
if and only if one of the following holds:
\begin{itemize}
\item either $\{b_1,b_2,b_3\}$ is a degenerate block and $v\in\{b_1,b_2,b_3\}$, 
\item or $b_1,b_2,b_3$ are mutually distinct and there is a proper block $B$ such that $b_1,b_2,b_3,v\in B$.
\end{itemize}
Moreover, for all $b_1,b_2,b_3,v,w\in V(G)$, we have
\[
G\models \logic{torso}^{(n)}(b_1,b_2,b_3,v,w)
\]
if and only if $G\models \logic{block}^{(n)}(b_1,b_2,b_3,v)$ and
$G\models \logic{block}^{(n)}(b_1,b_2,b_3,w)$ and $vw$ is an edge of
the torso of the block determined by $b_1,b_2,b_3$.
\end{lemma}

\begin{proof}
It is easy to express in $\LC{O(1)}{O(\log n)}$ that $\{b_1,b_2,b_3\}$ is a degenerate block. For proper blocks, we use Lemmas~\ref{lem:block:characterisation} and \ref{lem:block2}.
\end{proof}

As an immediate consequence, we obtain a formula to define a block separator.

\begin{corollary}\label{cor:blocksep}
For all $n \in \Nat$, there exists a $\LC{O(1)}{O(\log n)}$-formula $\logic{blocksep}^{(n)}(x_1,x_2)$ such that for all graphs $G$ of order at most $n$ and all $s_1,s_2\in V(G)$, we have
\[
G\models \logic{blocksep}^{(n)}(s_1,s_2)
\]
if and only if $\{s_1,s_2\}$ is a block separator of $G$.
\end{corollary}

We are ready to define the formula that yields the decomposition from Lemma \ref{lem:tree:dec:log:height}.

\begin{lemma}\label{lem:dec:logic}
For all $h\ge 0$, $n\ge 1$, there is a $\LC{O(1)}{O(h+\log n)}$-formula
$\logic{dec}^{(n)}_h(x_i^j,y_{k}\mid i\in[4],j\in[3],k\in[6])$ such that the
following holds.
Let $G$ be a graph of order $|G|\le n$ and $b_i^j,s_{k}\in V(G)$ for
$i\in[4],j\in[3],k\in [6]$ (not necessarily distinct). Then
\[
G\models \logic{dec}^{(n)}_h(b_i^j,s_{k}\mid i\in[4],j\in[3],k\in[6])
\]
if and only if the following conditions are satisfied.
\begin{eroman}
\item\label{it:dd1} For all $i\in [4]$, either $B_i \coloneqq
\{b_i^1,b_i^2,b_i^3\}$ is a block separator or $B_i \coloneqq \{b_i^1,b_i^2,b_i^3\}$ is a degenerate
block or $b_i^1,b_i^2,b_i^3$ are mutually distinct and there is
a (unique) block $B_i$ that contains $b_i^1,b_i^2,b_i^3$.

Let $B\coloneqq B_1\cup\ldots\cup B_4$.
\item\label{it:dd2} $S\coloneqq\{s_1,\ldots,s_6\}\subset B$.
\item\label{it:dd3} There is a (unique) connected component $C$ of
$G\setminus S$ such that $B\subseteq S\cup V(C)$.
\item\label{it:dd4} The induced subgraph $G[S\cup V(C)]$ has a rooted tree decomposition
$(T^*,\beta^*)$ of height at most $h$ with $B=\bag^*(r^*)$ for the root $r^*$ of $T^*$.
\item\label{it:dd5} The tree decomposition $(T^*,\beta^*)$ satisfies Conditions
\eqref{it:tdl2}--\eqref{it:tdl6} of Lemma~\ref{lem:tree:dec:log:height}, where all
blocks are blocks of the graph $G$  (rather than of the subgraph $G[S\cup C]$).
\end{eroman}
\end{lemma}

\begin{proof}
  The proof is by induction on $h \geq 0$.

  However, before we begin the induction, we observe that using
  Lemma~\ref{lem:block} and Corollary~\ref{cor:blocksep}, we can write a formula
  in the variables $x_i^j$ that expresses Condition \eqref{it:dd1}. It
  is straightforward to express Condition \eqref{it:dd2}, and, again using
  Lemma~\ref{lem:block}, to express Condition \eqref{it:dd3}. So in
  the induction, we will focus on Conditions \eqref{it:dd4} and \eqref{it:dd5}.

  For the base case $h=0$, we observe that a decomposition of
  height $0$ consists of a single node that covers the whole graph. So
  we need to express that for the component $C$ we obtain in \eqref{it:dd3}, we have $V(C)\cup S = B$. Then the $1$-node tree
  decomposition of $G[B]$ satisfies Conditions \eqref{it:dd4} and \eqref{it:dd5}.
  
  For a $1$-node decomposition,
  Conditions \eqref{it:tdl3} and \eqref{it:tdl6} of
  Lemma~\ref{lem:tree:dec:log:height} are void, and Condition
  \eqref{it:tdl2} of Lemma~\ref{lem:tree:dec:log:height} follows from Condition \eqref{it:dd1} of this lemma.

  For the inductive step $h\to h+1$, suppose we have a graph $G$ and
  elements $b_i^j$, $s_k$ satisfying Conditions \eqref{it:dd1}--\eqref{it:dd3}
  for suitable sets $B,S,C$. It suffices to express that for each
  connected component $C'$ of $G[S\cup V(C)]\setminus B$, we can find a decomposition of height $h$ that covers $C'$ and attaches to
  $B$ in a way that satisfies the conditions of
  Lemma~\ref{lem:tree:dec:log:height}.

  So let $G'\coloneqq G[S\cup V(C)]$, and let 
  $C'$ be a connected component of $G'\setminus B$. Let $S'\coloneqq
  N_G(C')$. If $|S'|> 6$, then there definitely is no decomposition
  with the desired properties. Suppose that $S'=\{s_1',\ldots,s_6'\}$. Then, if there are
  $b_i'^j\in S'\cup V(C')$ such that $G\models
  \logic{dec}_h^{(n)}(b_i'^j,s_k'\mid i\in[4],j\in[3],k\in[6])$, 
  the desired decomposition that covers $C'$ exists by the
  induction hypothesis. If this is the case for all $C'$, we can
  combine the decompositions to form the desired decomposition of
  $G'$. Conversely, if there is a decomposition of $G[S'\cup
  V(C')]$ of height $h$ in the sense of Lemma \ref{lem:tree:dec:log:height} such that for the root $u^*$, the bag
  $\beta^*(u^*)$ contains
  $S'$, then there are blocks or block separators $B_1',\ldots,B_4'$
  such that $\beta^*(u^*)=B_1'\cup\ldots\cup B_4'$. From the $B_i'$, we
  obtain $b_i'^j$ such that $G\models
  \logic{dec}_h^{(n)}(b_i'^j,s_k'\mid i\in[4],j\in[3],k\in[6])$,
  again by the induction hypothesis.

  To conclude, in addition to the subformulas taking care of
  Conditions \eqref{it:dd1}--\eqref{it:dd3}, the formula $\logic{dec}_{h+1}^{(n)}$
  must have a subformula stating that, for all connected components
  $C'$ of $G'\setminus B$, there exist $s_k'\in B$ for $k\in[6]$ and
  $b_i'^j\in S'\cup V(C')$ for
  $i\in[4]$, $j\in[3]$ such that
  $\{s_1',\ldots,s_6'\}=N_G(C')$ and
  $\logic{dec}_h^{(n)}(b_i'^j,s_k'\mid i\in[4],j\in[3],k\in[6])$ 
  holds.

  Note that, in each step $h\to h+1$ of the induction, we need to use
  formulas of quantifier depth $O(\log n)$ to express the desired
  connectivity conditions, for example to speak about components $C'$,
  and to express that the $b_i^j$ define blocks. However, the formula
  $\logic{dec}_h^{(n)}$ occurs only in the scope of constantly many
  (19, to be precise) quantifiers ranging over an element of the
  component(s) $C'$ and the $b_i'^j,s_k'$. Thus, overall, the
  quantifier depth will be $O(h)+O(\log n)$.
\end{proof}

\section{Canonisation}\label{sec:reduction}

In this section, we finally prove Theorems \ref{thm:main} and $\ref{thm:main:logical}$. By the logical characterisation of the WL algorithm given in Theorem \ref{thm:quantdepth}, we obtain Theorem \ref{thm:main} as a corollary from Theorem $\ref{thm:main:logical}$, which we prove below. 

In the following, for a graph $G$ and a list of vertices $v_1, \dots, v_\ell \in V(G)$, we denote by $(G, v_1, \dots, v_\ell)$ the graph $G$ with \emph{individualised} vertices $v_1, \dots, v_\ell$. That is, $(G, v_1, \dots, v_\ell)$ and $(G', v'_1, \dots, v'_{\ell'})$ have the same isomorphism type if and only if $\ell = \ell'$ and there is an isomorphism from $G$ to $G'$ that maps $v_i$ to $v'_i$ for every $i \in [\ell]$.

\begin{figure}
\centering
\includegraphics[width=7.8cm]{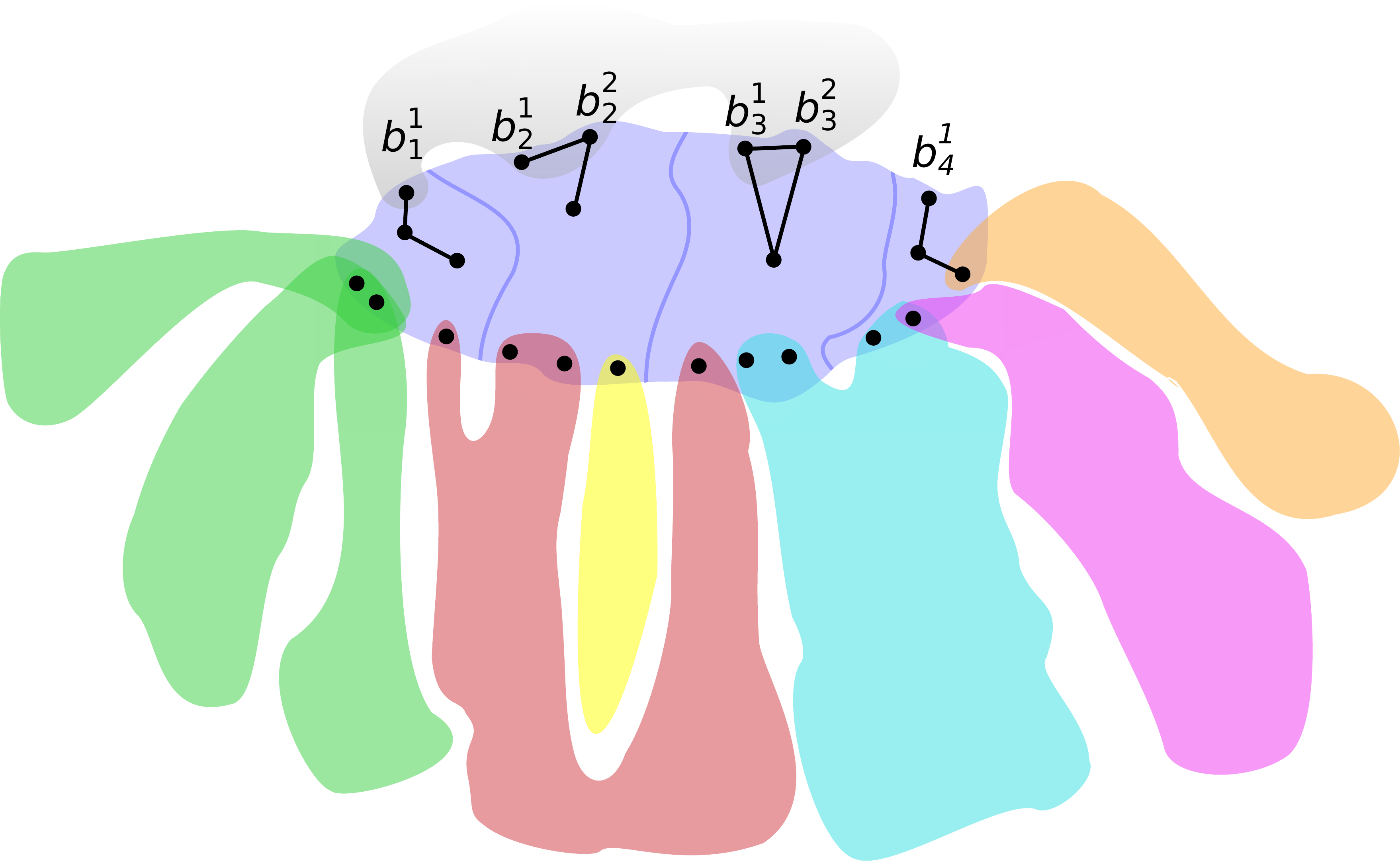}
\caption{A (simplified) schematic visualisation of the rooted tree decomposition $(T^*,\beta^*)$ from Lemma \ref{lem:tree:dec:log:height}. For simplicity, all $B_i$ in the bag of the purple node are depicted as distinct proper blocks.}
\label{fig:drawing}
\end{figure}

\begin{lemma}\label{lem:iso:formel}
For all $h\ge 0$, $n\ge 1$ and all connected planar graphs $G$ of order $|G|\le n$, and all $b_i^j,s_{k}\in V(G)$ for
$i\in[4],j\in[3],k\in [6]$ (not necessarily distinct) such that 
\[
G\models \logic{dec}^{(n)}_h(b_i^j,s_{k}\mid i\in[4],j\in[3],k\in[6]),
\]
there is a $\LC{O(1)}{O(h+\log n)}$-formula
$\logic{iso}^{(n)}_h(x_i^j,y_{k}\mid i\in[4],j\in[3],k\in[6])$ (which depends on the $b_i^j$ and the $s_k$) such that the
following holds.
Let $H$ be a connected graph of order $|H|\le n$ and $b_i'^j,s'_{k}\in V(H)$ for
$i\in[4],j\in[3],k\in [6]$ (not necessarily distinct) and assume $H\models \logic{dec}^{(n)}_h(b_i'^j,s'_{k}\mid i\in[4],j\in[3],k\in[6])$. Then
\[
H\models \logic{iso}^{(n)}_h(b_i'^j,s'_{k}\mid i\in[4],j\in[3],k\in[6])
\]
if and only if for the connected components $C_G, C_H$ that Lemma \ref{lem:dec:logic} yields for $G$ and $H$, it holds that $\big(H[\{s'_1,\dots,s'_6\} \cup V(C_H)],(b_i'^j,s'_k \mid i\in[4],j\in[3],k\in[6])\big) \cong \big(G[\{s_1,\dots,s_6\} \cup V(C_G)],(b_i^j,s_k \mid i\in[4],j\in[3],k\in[6])\big)$.
\end{lemma}

\begin{proof}
For the following arguments, see also Figure \ref{fig:drawing} for a better intuition.

Let $n \in \mathbb{N}$ and let $G$ be a connected planar graph with $|G| \leq n$. The proof is by induction on $h \geq 0$.

First, given a second connected graph $H$ of order at most $|G|$ that satisfies the $\formel{dec}_h^{(n)}$-formula, we can assume that the first four triplets of vertices form the same types of blocks and block separators (of corresponding sizes), respectively, in $H$ as in $G$, since otherwise we can distinguish the graphs using the formulas from Lemma \ref{lem:block} and Corollary~\ref{cor:blocksep}.

Note that there is a formula $\formel{bag}^{(n)}(x^1_1,\dots,x^3_{4},y) \in \LC{O(1)}{O(\log n)}$ such that for all graphs $H$ of order at most $n$ and all $b_1'^1, b_1'^2,b_1'^3, \dots, b_4'^1,b_4'^2,b_4'^3, v \in V(H)$, it holds that $H \models \formel{bag}^{(n)}(b_1'^1, \dots, b_4'^3,v)$ if and only if each set $\{b_i'^j \mid j \in [3]\}$ for $i \in [4]$ is a block separator $B_i$ or a degenerate block $B_i$ or contained in a proper block $B_i$ of $H$ and $v$ is in $B \coloneqq \bigcup_{i=1}^4 B_i$.

The case that $h = 0$ follows analogously as the formula for the isomorphism type of the root bag in the inductive step. We therefore focus on the inductive step. Assume that for every list of vertices $(b_i'^j,s'_k \mid i\in[4],j\in[3],k\in[6]) \in V(G)^{18}$, where \[
G\models \logic{dec}^{(n)}_h(b_i'^j,s'_{k} \mid i\in[4],j\in[3],k\in[6]),
\] 
there is a $\LC{O(1)}{O(h+ \log n)}$-formula 
\[\formel{iso}_{G,(b_i'^j,s'_k \mid i\in[4],j\in[3],k\in[6])}(x_1'^1,\dots,x_4'^3,y'_1,\dots,y'_6)\]
that defines the isomorphism type of $(G[\{s'_1,\dots,s'_6\} \cup V(C')],(b_i'^j,s'_k \mid i\in[4],j\in[3],k\in[6]))$, where $C'$ is the connected component from Parts \eqref{it:dd3}--\eqref{it:dd5} in Lemma \ref{lem:dec:logic}.

Let $(b_i^j,s_k \mid i\in[4],j\in[3],k\in[6]) \in V(G)^{18}$ be a list of vertices such that \[
G\models \logic{dec}^{(n)}_{h+1}(b_i^j,s_{k}\mid i\in[4],j\in[3],k\in[6]).
\]
For $B_1, B_2, B_3, B_4, B,C,S$ as described in Lemma
\ref{lem:dec:logic}, let $(T^*,\beta^*)$ be the rooted tree
decomposition from Condition \eqref{it:dd4} in Lemma
\ref{lem:dec:logic}. Let $r^*$ be the root of $T^*$. By Condition
\eqref{it:dd4} in Lemma \ref{lem:dec:logic}, it holds that
$\beta^*(r^*) = B = \bigcup_{i=1}^4 B_i$. 
Consider a $B_i$ with $|B_i| \geq 4$. Then $B_i$ is a proper block, in which, by Theorem \ref{thm:3conn}, we can find vertices $v_i^1, v_i^2, v_i^3$ such that for all $w \in B_i$, there is a $\LC{O(1)}{O(\log n)}$-formula $\logic{id}'_w(x_i^1,x_i^2,x_i^3,y)$ such that $G[[B_i]] \models \logic{id}'_w(v_i^1,v_i^2,v_i^3,w)$ and $G[[B_i]] \not\models \logic{id}'_w(v_i^1,v_i^2,v_i^3,w')$ for every $w'\in B_i \setminus \{v\}$. (In every $B_i$ with $|B_i| \leq 3$, such vertex-identifying formulas with four free variables exist trivially and they also identify the vertex the entire graph $G$.)

For simplicity, first assume that for all $i$ with $|B_i| \geq 4$, the vertex $v_i^j$ equals $b_i^j$ for $j \in [3]$. Then by replacing in $\logic{id}'_v(x_i^1,x_i^2,x_i^3,y)$ every subformula of the form $\exists^{\geq k} x \psi$ with $\exists^{\geq k} x (\psi \land \logic{block}^{(n)}(x_i^1,x_i^2,x_i^3,x))$ and every $E(x,y)$ with $\formel{torso}^{(n)}(x_i^1,x_i^2,x_i^3,x,y)$, we easily obtain for each $v \in B$ a $\LC{O(1)}{O(\log n)}$-formula $\logic{id}_v(x_1^1,\dots,x_4^3,y)$ with $\logic{id}_v[G,b_1^1,\dots,b_4^3,y] = \{v\}$. 

Now we can use these formulas to address each vertex
individually. More formally, we can define the edge relation of $G[B]$
by setting, for $v,w \in B$ with $v \neq w$, 
  \[
    \varphi_{v,w}(x,y)\coloneqq
    \begin{cases}
      E(x,y)&\text{if }vw \in E(G),\\
      \neg E(x,y)&\text{otherwise}.
    \end{cases}
  \]

Then the $\LC{O(1)}{O(\log n)}$-formula 
\begin{align*}
\formel{iso}_B(x_1^1, \dots, x_4^3) {{}\coloneqq}{} 
\mspace{-7mu}&\bigwedge_{v,w \in B} \mspace{-12mu}\exists^{=1} x \Big(\logic{id}_v(x_1^1,\dots,x_4^3,x) \land \exists^{=1} x'\big(\logic{id}_{w}(x_1^1,\dots,x_4^3,x') \land \varphi_{v,w}(x,x')\big)\mspace{-3mu}\Big) 
\\&\phantom{\text{ \ }}\land \neg \exists x \Big(\formel{bag}^{(n)}(x_1^1,\dots,x_4^3,x) \rightarrow \bigwedge_{w \in B} \neg \logic{id}_{w}(x_1^1,\dots,x_4^3,x)\Big) \land {}
\\&\mspace{-3mu}\bigwedge_{v\neq w \in B} \mspace{-12mu}\neg\exists x \Big(\logic{id}_v(x_1^1,\dots,x_4^3,x) \land \logic{id}_{w}(x_1^1,\dots,x_4^3,x)\Big)
\end{align*} 
defines the isomorphism type of $(G[B],b_1^1, \dots, b_4^3)$ (see the
purple bag in Figure \ref{fig:drawing}).

We now construct a formula that describes how the connected components of $G[S \cup V(C)] \setminus B$ are attached to $G[B]$. Let $G' \coloneqq G[S \cup V(C)]$. By Condition \eqref{it:tdl3} in Lemma \ref{lem:tree:dec:log:height}, for every connected component $C'$ of $G'\setminus B$, it holds that $|N_G(C')| \leq 6$ (see the coloured shapes attached to the purple one in Figure \ref{fig:drawing}). Hence, we iterate over all tuples $(s'_1, \dots, s'_6) \in B^6$: let $\CM^{s'_1, \dots, s'_6}$ be the multiset of isomorphism types of the graphs $(G[S' \cup C'],s'_1,\dots,s'_6)$, where $S' \coloneqq \{s'_i \mid i \in [6]\}$ and $C'$ is a connected component of $G'\setminus B$ with $N_G(C') = S'$. 

Since $G\models \logic{dec}^{(n)}_{h+1}(b_i^j,s_{k}\mid i\in[4],j\in[3],k\in[6])$, for every $(s'_1,\dots,s'_6) \in B^6$ and every connected component $C'$ of $G' \setminus B$ with $N_G(C') = \{s'_1,\dots,s'_6\}$, there exist vertices $(b_i'^j \mid i\in[4],j\in[3]) \in (S' \cup V(C'))^{12}$ such that 
\[
G[S'\cup V(C')]\models \logic{dec}^{(n)}_{h}(b_i'^j,s'_{k}\mid i\in[4],j\in[3],k\in[6]).
\]
So, by the induction hypothesis, there is a formula $\formel{iso}_M(x_1'^1,\dots,x_4'^3,y'_1,\dots,y'_6) \in \LC{O(1)}{O(h+ \log n)}$ for the isomorphism type $M$ of $({G[\{s'_1, \dots, s'_6\}\cup V(C')],(b_i'^j,s'_k \mid i\in[4],j\in[3],k\in[6]))}$. Note that by Condition \eqref{it:dd2} in Lemma \ref{lem:dec:logic}, at least one of the vertices $b_i'^j$ will lie outside $B$. Using the counting quantifiers, we can use the $\formel{iso}_M$ to make sure that every isomorphism type appears with the correct multiplicity. More precisely, we first group all components with equal isomorphism types. The fact that they are of the same size enables us to define their number (e.g.\ the three green shapes in Figure \ref{fig:drawing}).  This then allows us to build a formula $\formel{iso}'_\CM(y'_1,\dots,y'_6)$ which identifies the graph $(G[S' \cup \bigcup_{C' : N_G(C') = S'} V(C')], s'_1, \dots,s'_6)$ (where $\CM \coloneqq \CM^{s'_1, \dots, s'_6}$ and the $C'$ are connected components of $G' \setminus B$). Using the $\formel{dec}_h^{(n)}$-formula, we can turn $\formel{iso}'_\CM(y'_1,\dots,y'_6)$ into a formula $\formel{iso}_\CM(x_1^1,\dots,x_4^3,y_1,\dots,y_6,y'_1,\dots,y'_6)$ that ensures that $\formel{iso}_\CM(b_1^1,\dots,b_4^3,s_1,\dots,s_6,s'_1,\dots,s'_6)$ describes for $S' \coloneqq \{s'_1,\dots,s'_6\}$ the subgraph $(G[S' \cup \bigcup_{C' : N_G(C') = S'} V(C')], s'_1, \dots,s'_6)$, where the $C'$ are connected components of $G' \setminus B$, up to isomorphism.

Hence, it suffices to conjugate $\formel{iso}_B(x_1^1,\dots,x_4^3)$ with a conjunction over all $(s'_1, \dots, s'_6) \in B^6$ of the following formula 
\[\exists y'_1 \dots \exists y'_6 \left(\bigwedge_{i=1}^6 \formel{id}_{s'_i}(x_1^1, \dots, x_4^3,y'_i) \wedge \formel{iso}_\CM(x_1^1,\dots,x_4^3,y_1,\dots,y_6,y'_1,\dots,y'_6)\right),\]
where $\CM \coloneqq \CM^{s_1, \dots, s_6}$,
to obtain the desired $\formel{iso}_{G,(b_i^j,s_k \mid i\in[4],j\in[3],k\in[6])}(x_1^1,\dots,x_4^3,y_1,\dots,y_6)$.

We now consider the general case where it does not necessarily hold for all $i,j$ that $v_i^j = b_i^j$. We assume for notational simplicity that for all $i$, the $b_i^1,b_i^2,b_i^3$ define a block. It is easy to adapt the following construction to the situation that block separators are present.

We introduce one nested existential quantifier $\exists \tilde{x}_i^j$ for each of the $v_i^j$ so that our resulting formula $\logic{iso}_{G,(b_i^j,s_k \mid i\in[4],j\in[3],k\in[6])}(x_1^1,\dots,x_4^3,y_1,\dots,y_6)$ looks as follows:
\begin{align*}\textstyle
\exists \tilde{x}_1^1 \dots \exists \tilde{x}_4^3 \Bigg( &\bigwedge_{j=1}^3 \bigwedge_{i=1}^4 \formel{block}^{(n)}(x_i^1,x_i^2,x_i^3,\tilde{x}_i^j) \land \formel{iso}_B(\tilde{x}_1^1,\dots,\tilde{x}_4^3)
 \land {}\\&\mspace{-30mu}\bigwedge_{(s'_1,\dots,s'_6)\in B^6} \mspace{-15mu} \exists y'_1 \dots \exists y'_6 \Bigg(\bigwedge_{i=1}^6 \formel{id}_{s'_i}(\tilde{x}_1^1, \dots, \tilde{x}_4^3,y'_i) \wedge {}
 \\&\phantom{{}\mspace{-30mu}\bigwedge_{(s'_1,\dots,s'_6)\in B^6} \mspace{-15mu} \exists y'_1 \dots \exists y'_6 \Bigg(\bigwedge_{i=1}^6{}{}}\formel{iso}_\CM(x_1^1,\dots,x_4^3,y_1,\dots,y_6,y'_1,\dots,y'_6)\Bigg)\Bigg).
\end{align*}

The bounds on the quantifier depth and the number of variables follow similarly as in the proof of Lemma \ref{lem:dec:logic}.
\end{proof}

Applying Lemma \ref{lem:logdec}, we can deduce Theorem \ref{thm:main:logical}.

\begin{proof}[Proof of Theorem \ref{thm:main:logical}]
Let $n \in \mathbb{N}$ and let $G$ be a planar graph with order $|G| = n$. If $G$ is not connected, we construct one formula for each connected component of $G$ (as described in the following) and join them to obtain the identifying sentence.

So suppose $G$ is connected. Then by Lemma \ref{lem:tree:dec:log:height}, $G$ has a rooted tree decomposition $(T^*,\beta^*)$ of logarithmic height and adhesion at most $6$ for which every bag is a union of four (not necessarily distinct) blocks or block separators and also Condition \eqref{it:tdl6} of the lemma holds. Let $b_1^1, \dots, b_4^3$ be vertices that determine the blocks and block separators in the root bag $B$ of $(T^*,\beta^*)$. 

If there is a vertex $s \in B$ such that there is a unique connected component $C$ of $G \setminus \{s\}$ with $B \subseteq \{s\} \cup V(C)$, then there are vertices $b_i^j, s_k$ for $i \in [4], j \in [3], k \in [6]$ (e.g.\ $s_k = s$ for all $k$) such that $G$ satisfies $\formel{dec}^{(n)}_{2 \log|G|}(b_i^j,s_{k}\mid i\in[4],j\in[3],k\in[6])$. Then the sentence 
\[\exists x_1^1 \dots \exists x_4^3 \exists y_1 \dots \exists y_6 \formel{iso}_{G,(b_i^j,s_k \mid i\in[4],j\in[3],k\in[6])}(x_1^1,\dots,x_4^3,y_1,\dots,y_6)\]
identifies $G$, where $\formel{iso}_{G,(b_i^j,s_k \mid i\in[4],j\in[3],k\in[6])}$ is the formula from Lemma \ref{lem:iso:formel}. 

Otherwise, let $s \in B$ be a vertex such that $G \setminus \{s\}$ has multiple connected components $C_i$ and let $G_i \coloneqq G[V(C_i) \cup \{s\}]$. Then the restriction of $(T^*,\beta^*)$ to each $G_i$ still satisfies the conditions of Lemma \ref{lem:tree:dec:log:height}, because the block structure of $G_i$ is just the block structure induced by $G$ on $V(G_i)$ (that is, the blocks of $G_i$ are precisely those blocks of $G$ contained in $V(G_i)$, and similarly for the block separators). This yields by Lemma \ref{lem:iso:formel} an identifying formula $\varphi_i(y)$ for each $(G_i,s)$, which we can join by isomorphism type of $(G_i,s)$ to obtain an identifying sentence.
\end{proof}

We can directly deduce Theorem \ref{thm:main}.

\begin{proof}[Proof of Theorem \ref{thm:main}]
The theorem follows from Theorems \ref{thm:main:logical} and \ref{thm:quantdepth}.
\end{proof}

\section{Conclusion}
We prove that planar graphs are identified by the WL algorithm with constant dimension in a logarithmic number of
iterations, thereby completing a project started by Verbitsky fourteen years
ago with his proof of the same result in the special case of 3-connected planar
graphs. Our proof is based on the careful analysis of a novel
logarithmic-depth decomposition of graphs into their 3-connected
components. 

It is unclear which dimension of the WL algorithm is necessary to
identify planar graphs in logarithmically many iterations and if there
is a (provable) trade-off between dimension and iteration number. This
is not only interesting for planar graphs, and many questions remain
open.

We leave it as another interesting open project whether our
result can be extended to graph classes of bounded genus. As it
stands, our proof heavily relies on properties of 3-connected planar
graphs that are not shared by 3-connected graphs of higher genus. Similarly, we pose as a challenge to find good bounds on the iteration number of the WL algorithm on other parameterised graph classes, such as graphs with a certain excluded minor or graphs of bounded rank width.  

\bibliography{main}

\begin{thebibliography}{10}

\bibitem{ahmkermlanat13}
B.~Ahmadi, K.~Kersting, M.~Mladenov, and S.~Natarajan.
\newblock Exploiting symmetries for scaling loopy belief propagation and
  relational training.
\newblock {\em Machine Learning Journal}, 92(1):91--132, 2013.
\newblock \href {https://doi.org/10.1007/s10994-013-5385-0}
  {\path{doi:10.1007/s10994-013-5385-0}}.

\bibitem{atsman13}
A.~Atserias and E.~N. Maneva.
\newblock {Sherali-Adams} relaxations and indistinguishability in counting
  logics.
\newblock {\em {SIAM} Journal on Computing}, 42(1):112--137, 2013.
\newblock \href {https://doi.org/10.1137/120867834}
  {\path{doi:10.1137/120867834}}.

\bibitem{atsoch18}
A.~Atserias and J.~Ochremiak.
\newblock Definable ellipsoid method, sums-of-squares proofs, and the
  isomorphism problem.
\newblock In {\em Proceedings of the 33rd Annual {ACM/IEEE} Symposium on Logic
  in Computer Science (LICS '18)}, pages 66--75, 2018.
\newblock \href {https://doi.org/10.1145/3209108.3209186}
  {\path{doi:10.1145/3209108.3209186}}.

\bibitem{bab16}
L.~Babai.
\newblock Graph isomorphism in quasipolynomial time.
\newblock In {\em Proceedings of the 48th Annual {ACM} Symposium on Theory of
  Computing ({STOC} '16)}, pages 684--697, 2016.
\newblock \href {https://doi.org/10.1145/2897518.2897542}
  {\path{doi:10.1145/2897518.2897542}}.

\bibitem{caifurimm92}
J.~Cai, M.~F{\"u}rer, and N.~Immerman.
\newblock An optimal lower bound on the number of variables for graph
  identification.
\newblock {\em Combinatorica}, 12:389--410, 1992.
\newblock \href {https://doi.org/10.1007/BF01305232}
  {\path{doi:10.1007/BF01305232}}.

\bibitem{chepon19}
G.~Chen and I.~Ponomarenko.
\newblock Lectures on coherent configurations.
\newblock Lecture notes available at
  \url{http://www.pdmi.ras.ru/~inp/ccNOTES.pdf}, 2019.

\bibitem{darliffsakmar04}
P.~T. Darga, M.~H. Liffiton, K.~A. Sakallah, and I.~L. Markov.
\newblock Exploiting structure in symmetry detection for {CNF}.
\newblock In {\em Proceedings of the 41st Design Automation Conference ({DAC}
  '04)}, pages 530--534. {ACM}, 2004.
\newblock \href {https://doi.org/10.1145/996566.996712}
  {\path{doi:10.1145/996566.996712}}.

\bibitem{delgrorat18}
H.~Dell, M.~Grohe, and G.~Rattan.
\newblock Lov{\'{a}}sz meets {W}eisfeiler and {L}eman.
\newblock In {\em Proceedings of the 45th International Colloquium on Automata,
  Languages, and Programming ({ICALP} '18)}, pages 40:1--40:14, 2018.
\newblock \href {https://doi.org/10.4230/LIPIcs.ICALP.2018.40}
  {\path{doi:10.4230/LIPIcs.ICALP.2018.40}}.

\bibitem{die16}
R.~Diestel.
\newblock {\em Graph Theory}.
\newblock Springer Verlag, 5th edition, 2016.

\bibitem{dvo10}
Z.~Dvor{\'{a}}k.
\newblock On recognizing graphs by numbers of homomorphisms.
\newblock {\em Journal of Graph Theory}, 64(4):330--342, 2010.
\newblock \href {https://doi.org/10.1002/jgt.20461}
  {\path{doi:10.1002/jgt.20461}}.

\bibitem{elbjaktan10}
M.~Elberfeld, A.~Jakoby, and T.~Tantau.
\newblock Logspace versions of the theorems of {Bodlaender} and {Courcelle}.
\newblock In {\em Proceedings of the 51st Annual IEEE Symposium on Foundations
  of Computer Science (FOCS '10)}, pages 143--152, 2010.
\newblock \href {https://doi.org/10.1109/FOCS.2010.21}
  {\path{doi:10.1109/FOCS.2010.21}}.

\bibitem{evdpontin00}
S.~Evdokimov, I.~N. Ponomarenko, and G.~Tinhofer.
\newblock Forestal algebras and algebraic forests (on a new class of weakly
  compact graphs).
\newblock {\em Discrete Mathematics}, 225(1-3):149--172, 2000.
\newblock \href {https://doi.org/10.1016/S0012-365X(00)00152-7}
  {\path{doi:10.1016/S0012-365X(00)00152-7}}.

\bibitem{gro98a}
M.~Grohe.
\newblock Fixed-point logics on planar graphs.
\newblock In {\em Proceedings of the 13th IEEE Symposium on Logic in Computer
  Science (LICS '98)}, pages 6--15, 1998.
\newblock \href {https://doi.org/10.1109/LICS.1998.705639}
  {\path{doi:10.1109/LICS.1998.705639}}.

\bibitem{gro00}
M.~Grohe.
\newblock Isomorphism testing for embeddable graphs through definability.
\newblock In {\em Proceedings of the 32nd ACM Symposium on Theory of Computing
  (STOC '00)}, pages 63--72, 2000.
\newblock \href {https://doi.org/10.1145/335305.335313}
  {\path{doi:10.1145/335305.335313}}.

\bibitem{gro17}
M.~Grohe.
\newblock {\em Descriptive Complexity, Canonisation, and Definable Graph
  Structure Theory}, volume~47 of {\em Lecture Notes in Logic}.
\newblock Cambridge University Press, 2017.
\newblock \href {https://doi.org/10.1017/9781139028868}
  {\path{doi:10.1017/9781139028868}}.

\bibitem{gro21}
M.~Grohe.
\newblock The logic of graph neural networks.
\newblock In {\em Proceedings of the 36th ACM-IEEE Symposium on Logic in
  Computer Science (LICS '21))}, 2021.
\newblock arXiv version at \url{https://arxiv.org/abs/2104.14624}.

\bibitem{grokie19}
M.~Grohe and S.~Kiefer.
\newblock A linear upper bound on the {W}eisfeiler-{L}eman dimension of graphs
  of bounded genus.
\newblock In {\em Proceedings of the 46th International Colloquium on Automata,
  Languages, and Programming (ICALP '19)}, pages 117:1--117:15, 2019.
\newblock \href {https://doi.org/10.4230/LIPIcs.ICALP.2019.117}
  {\path{doi:10.4230/LIPIcs.ICALP.2019.117}}.

\bibitem{gromar99}
M.~Grohe and J.~Mari{\~{n}}o.
\newblock Definability and descriptive complexity on databases of bounded
  tree-width.
\newblock In {\em Proceedings of the 7th International Conference on Database
  Theory (ICDT '99)}, volume 1540 of {\em Lecture Notes in Computer Science},
  pages 70--82. Springer, 1999.
\newblock \href {https://doi.org/10.1007/3-540-49257-7_6}
  {\path{doi:10.1007/3-540-49257-7_6}}.

\bibitem{groneu19}
M.~Grohe and D.~Neuen.
\newblock Canonisation and definability for graphs of bounded rank width.
\newblock In {\em Proceedings of the 34th Annual {ACM/IEEE} Symposium on Logic
  in Computer Science (LICS '19)}, pages 1--13, 2019.
\newblock \href {https://doi.org/10.1109/LICS.2019.8785682}
  {\path{doi:10.1109/LICS.2019.8785682}}.

\bibitem{groott15}
M.~Grohe and M.~Otto.
\newblock Pebble games and linear equations.
\newblock {\em Journal of Symbolic Logic}, 80(3):797--844, 2015.
\newblock \href {https://doi.org/10.1017/jsl.2015.28}
  {\path{doi:10.1017/jsl.2015.28}}.

\bibitem{grover06}
M.~Grohe and O.~Verbitsky.
\newblock Testing graph isomorphism in parallel by playing a game.
\newblock In {\em Proceedings of the 33rd International Colloquium on Automata,
  Languages and Programming (ICALP '06)}, pages 3--14, 2006.
\newblock \href {https://doi.org/10.1007/11786986\_2}
  {\path{doi:10.1007/11786986\_2}}.

\bibitem{immlan90}
N.~Immerman and E.~Lander.
\newblock Describing graphs: A first-order approach to graph canonization.
\newblock In {\em Complexity theory retrospective}, pages 59--81.
  Springer-Verlag, 1990.

\bibitem{junkas07}
T.~A. Junttila and P.~Kaski.
\newblock Engineering an efficient canonical labeling tool for large and sparse
  graphs.
\newblock In {\em Proceedings of the 9th Workshop on Algorithm Engineering and
  Experiments ({ALENEX} '07)}. {SIAM}, 2007.
\newblock \href {https://doi.org/10.1137/1.9781611972870.13}
  {\path{doi:10.1137/1.9781611972870.13}}.

\bibitem{kie20}
S.~Kiefer.
\newblock The {W}eisfeiler-{L}eman algorithm: An exploration of its power.
\newblock {\em {ACM} {SIGLOG} News}, 7(3):5--27, 2020.
\newblock \href {https://doi.org/10.1145/3436980.3436982}
  {\path{doi:10.1145/3436980.3436982}}.

\bibitem{kiemck20}
S.~Kiefer and B.~D. McKay.
\newblock The iteration number of {C}olour {R}efinement.
\newblock In {\em Proceedings of the 47th International Colloquium on Automata,
  Languages, and Programming ({ICALP} '20)}, volume 168 of {\em LIPIcs}, pages
  73:1--73:19. Schloss Dagstuhl -- Leibniz-Zentrum f{\"{u}}r Informatik, 2020.
\newblock \href {https://doi.org/10.4230/LIPIcs.ICALP.2020.73}
  {\path{doi:10.4230/LIPIcs.ICALP.2020.73}}.

\bibitem{kieneu19}
S.~Kiefer and D.~Neuen.
\newblock The power of the {W}eisfeiler-{L}eman algorithm to decompose graphs.
\newblock In {\em Proceedings of the 44th International Symposium on
  Mathematical Foundations of Computer Science (MFCS '19)}, volume 138 of {\em
  Leibniz International Proceedings in Informatics (LIPIcs)}, pages
  45:1--45:15. Schloss Dagstuhl -- Leibniz-Zentrum für Informatik, 2019.
\newblock \href {https://doi.org/10.4230/LIPIcs.MFCS.2019.45}
  {\path{doi:10.4230/LIPIcs.MFCS.2019.45}}.

\bibitem{kieponschwei19}
S.~Kiefer, I.~Ponomarenko, and P.~Schweitzer.
\newblock The {W}eisfeiler-{L}eman dimension of planar graphs is at most 3.
\newblock {\em J. {ACM}}, 66(6):44:1--44:31, 2019.
\newblock \href {https://doi.org/10.1145/3333003} {\path{doi:10.1145/3333003}}.

\bibitem{kieschwei19}
S.~Kiefer and P.~Schweitzer.
\newblock Upper bounds on the quantifier depth for graph differentiation in
  first-order logic.
\newblock {\em Log. Methods Comput. Sci.}, 15(2), 2019.
\newblock \href {https://doi.org/10.23638/LMCS-15(2:19)2019}
  {\path{doi:10.23638/LMCS-15(2:19)2019}}.

\bibitem{KoblerV08}
J.~K{\"{o}}bler and O.~Verbitsky.
\newblock From invariants to canonization in parallel.
\newblock In {\em Proceedings of the 3rd International Computer Science
  Symposium in Russia ({CSR} '08)}, volume 5010 of {\em Lecture Notes in
  Computer Science}, pages 216--227. Springer, 2008.
\newblock \href {https://doi.org/10.1007/978-3-540-79709-8_23}
  {\path{doi:10.1007/978-3-540-79709-8_23}}.

\bibitem{lichponschwei19}
M.~Lichter, I.~Ponomarenko, and P.~Schweitzer.
\newblock Walk refinement, walk logic, and the iteration number of the
  {Weisfeiler}-{Leman} algorithm.
\newblock In {\em Proceedings of the 34th Annual {ACM/IEEE} Symposium on Logic
  in Computer Science ({LICS} '19)}, pages 1--13. {IEEE}, 2019.
\newblock \href {https://doi.org/10.1109/LICS.2019.8785694}
  {\path{doi:10.1109/LICS.2019.8785694}}.

\bibitem{mck81}
B.~D. McKay.
\newblock Practical graph isomorphism.
\newblock {\em Congressus Numerantium}, 30:45--87, 1981.

\bibitem{mckaypip14}
B.~D. McKay and A.~Piperno.
\newblock Practical graph isomorphism, {II}.
\newblock {\em J. Symb. Comput.}, 60:94--112, 2014.
\newblock \href {https://doi.org/10.1016/j.jsc.2013.09.003}
  {\path{doi:10.1016/j.jsc.2013.09.003}}.

\bibitem{morritfey+19}
C.~Morris, M.~Ritzert, M.~Fey, W.~Hamilton, J.~E. Lenssen, G.~Rattan, and
  M.~Grohe.
\newblock {Weisfeiler} and {Leman} go neural: Higher-order graph neural
  networks.
\newblock In {\em Proceedings of the 33rd AAAI Conference on Artificial
  Intelligence}, 2019.
\newblock \href {https://doi.org/10.1609/aaai.v33i01.33014602}
  {\path{doi:10.1609/aaai.v33i01.33014602}}.

\bibitem{sheschlee+11}
N.~Shervashidze, P.~Schweitzer, E.~J. van Leeuwen, K.~Mehlhorn, and K.~M.
  Borgwardt.
\newblock Weisfeiler-{L}ehman graph kernels.
\newblock {\em Journal of Machine Learning Research}, 12:2539--2561, 2011.

\bibitem{tut84}
W.~T. Tutte.
\newblock {\em Graph Theory}.
\newblock Addison-Wesley, 1984.

\bibitem{ver07}
O.~Verbitsky.
\newblock Planar graphs: Logical complexity and parallel isomorphism tests.
\newblock In {\em Proceedings of the 24th Annual Symposium on Theoretical
  Aspects of Computer Science (STACS '07)}, pages 682--693, 2007.
\newblock \href {https://doi.org/10.1007/978-3-540-70918-3\_58}
  {\path{doi:10.1007/978-3-540-70918-3\_58}}.

\bibitem{weilem68}
B.~Weisfeiler and A.~Leman.
\newblock The reduction of a graph to canonical form and the algebra which
  appears therein.
\newblock {\em NTI, Series 2}, 1968.
\newblock English translation by G.~Ryabov available at
  \url{https://www.iti.zcu.cz/wl2018/pdf/wl_paper_translation.pdf}.

\bibitem{whi32}
H.~Whitney.
\newblock Congruent graphs and the connectivity of graphs.
\newblock {\em American Journal of Mathematics}, 54:150--168, 1932.

\bibitem{xuhulesjeg19}
K.~Xu, W.~Hu, J.~Leskovec, and S.~Jegelka.
\newblock How powerful are graph neural networks?
\newblock In {\em Proceedings of the 7th International Conference on Learning
  Representations (ICLR '19)}, 2019.

\end{thebibliography}
\end{document}